\documentclass[letterpaper, cleveref, autoref, thm-restate]{article}

\title{Improved Distributed Algorithms for Random Colorings}
\author{Charlie Carlson\thanks{Department of Computer Science, University of California, Santa Barbara. Email: \{charlieannecarlson,vigoda\}@ucsb.edu.  Research supported in part by NSF grant CCF-2147094.} \and Daniel Frishberg\thanks{Department of Computer Science and Software Engineering, California Polytechnic State University, San Luis Obispo. Email: dfrishbe@calpoly.edu} \and Eric Vigoda$^{*}$}

\usepackage[ruled, linesnumbered]{algorithm2e}
\usepackage[margin=1in]{geometry}
\usepackage{amsmath,amsfonts,amssymb,amsthm}
\usepackage{graphicx}
\usepackage{enumerate}
\usepackage{bbm}
\usepackage{verbatim,nameref}
\usepackage{hyperref,color}
\usepackage[capitalize,nameinlink]{cleveref}
\usepackage[dvipsnames]{xcolor}
\hypersetup{
	colorlinks=true,
	pdfpagemode=UseNone,
    citecolor=OliveGreen,
    linkcolor=NavyBlue,
    urlcolor=Magenta,
	pdfstartview=FitW
}
\usepackage{appendix}
\crefname{appsec}{Appendix}{Appendices}
\usepackage{tikz}

\theoremstyle{plain}
\newtheorem{theorem}{Theorem}[section]
\newtheorem{lemma}[theorem]{Lemma}

\theoremstyle{definition}
\newtheorem{definition}[theorem]{Definition}
\newtheorem{observation}[theorem]{Observation}

\newtheorem*{assumption*}{Assumption}

\theoremstyle{remark}
\newtheorem{remark}[theorem]{Remark}

\crefname{lemma}{Lemma}{Lemmas}
\crefname{theorem}{Theorem}{Theorems}
\crefname{definition}{Definition}{Definitions}
\crefname{fact}{Fact}{Facts}
\crefname{claim}{Claim}{Claims}
\crefname{proposition}{Proposition}{Propositions}

\usepackage{color}

\newcommand{\poly}{\mathrm{poly}}
\newcommand{\dist}{\mathrm{dist}}
\newcommand{\distflip}{\mathrm{dist}}

\newcommand{\pres}{\operatorname{pres}}

\newcommand{\eps}{\varepsilon}
\newcommand{\indicator}{\mathbf{1}}

\newcommand{\R}{\mathbb{R}}

\newcommand{\EE}{\mathcal{E}}

\renewcommand{\SS}{\mathcal{S}}

\newcommand{\MCflip}{\mathcal{MC}_{\mathrm{flip}}}
\newcommand{\Tmix}{T_{\mathrm{mix}}}
\newcommand{\Tcouple}{T_{\mathrm{couple}}}

\newcommand{\MC}{\mathcal{MC}}
\newcommand{\wHamming}{\mathcal{H}}

\def\Prob#1{{\mathbf{Pr}\left({#1}\right)}}
\def\Exp#1{{\mathbf{E}\left({#1}\right)}}

\def\ProbCond#1#2{{\mathbf{Pr}\left({#1} \mid {#2} \right)}}

\def\ExpCond#1#2{{\mathbf{E}\left({#1} \mid {#2} \right)}}

\begin{document}
\maketitle

\abstract{

We study distributed versions of Markov Chain Monte Carlo (MCMC) algorithms for generating random $k$-colorings of an input graph with maximum degree $\Delta$.
In the sequential setting, the Glauber dynamics is the simple MCMC algorithm which updates the color at a randomly chosen vertex in each step.  Fischer and Ghaffari (2018), and independently Feng, Hayes, and Yin (2018), presented a parallel and distributed version of the Glauber dynamics which converges in $O(\log{n})$ rounds for $k>(2+\eps)\Delta$ for any $\eps>0$. 
We present the distributed flip dynamics and prove $O(n\log{n})$ mixing for $k>(11/6-\delta)\Delta$ for a fixed $\delta>0$.  Our new Markov chain is a generalization of the distributed Glauber dynamics previously analyzed, and is a parallel and distributed version of the more general flip dynamics considered in the sequential setting which recolors local maximal two-colored components in each step.   While the distributed Glauber dynamics and the sequential flip dynamics are symmetric Markov chains, and hence their stationary distribution is uniformly distributed over colorings, our distributed flip dynamics is not symmetric and hence the stationary distribution is unclear.

}

\maketitle
\section{Introduction}

This paper presents parallel and distributed algorithms for sampling from high-dimensional distributions.  An important application is sampling from the equilibrium distribution of a graphical model.  The equilibrium distribution is often known as the Gibbs or Boltzmann distribution, and efficient sampling from the Gibbs/Boltzmann distribution is a key step for Bayesian inference~\cite{KF:book,Murphy}. 

Our focus is algorithms in the LOCAL model for the $k$-colorings problem. The $k$-colorings problem is a graphical model of particular combinatorial interest and has played an important role in the development of algorithmic sampling techniques with provable guarantees. The LOCAL model is a standard model of distributed computation due to Linial~\cite{linial-local}. 

In the LOCAL model, the input to a problem is generally a graph~$G = (V, E)$. Each vertex is identified with a processor and is assigned a unique identifier. In each \emph{round} of an algorithm, each vertex is allowed to send an unbounded amount of information (a \emph{message}) to each of its neighbors, and may perform an unbounded amount of computation locally.

For an input graph $G=(V,E)$ and integer $k\geq 2$, let $\Omega$ denote the proper (vertex) $k$-colorings of $G$, namely $\Omega=\{\sigma:V\rightarrow \{1,\dots,k\}: \mbox{ for all } (v,w)\in E, \sigma(v)\neq\sigma(w)\}$ is the collection of assignments of $k$ colors to the vertices so that neighboring vertices receive different colors.  The associated Gibbs distribution $\mu$ is the uniform distribution over $\Omega$, the space of proper $k$-colorings.  

Under mild conditions on $G$ (e.g., triangle-free~\cite{BBCK}), the number of $k$-colorings is exponentially large, i.e., $|\Omega|=\exp(\Omega(n))$.  Nevertheless, our goal is to sample from $\mu$, the uniform distribution over this exponentially large set, in time $\poly(n)$, and ideally in time $O(n\log{n})$.  Furthermore,  in the distributed setting our goal is to generate samples ideally in time $O(\log{n})$. 
Unfortunately, the algorithm we present samples from a distribution which is not the uniform distribution.

A common technique for sampling from the Gibbs distribution in a wide range of scientific fields is the {\em Markov Chain Monte Carlo (MCMC)} method.  The simplest example of an MCMC algorithm is the Glauber dynamics, also known as the Gibbs sampler.  

Consider an input graph $G=(V,E)$ with maximum degree $\Delta$, and $k\geq\Delta+2$.  The Glauber dynamics updates the color of a randomly chosen vertex in each step.  In particular, from a coloring $X_t\in\Omega$ at time $t$, the transitions $X_t\rightarrow X_{t+1}$ of the Glauber dynamics work as follows.  We choose a random vertex~$v$ uniformly at random from~$V$, and a color~$c$ uniformly at random from the set of colors $\{1,\dots,k\}$.   If no neighbor of~$v$ has color $c$ in the current coloring $X_t$, i.e., $c\not\in X_t(N(v))$ where $N(v)$ are the neighbors of vertex $v$, then we recolor $v$ as $X_{t+1}(v)=c$ and otherwise we set $X_{t+1}(v)=X_t(v)$.  For all other vertices $w\neq v$ we set $X_{t+1}(w)=X_t(w)$.  This corresponds to the Metropolis version of the Glauber dynamics.  Alternatively one can choose the color $c$ uniformly from $\{1,\dots,k\}\setminus X_t(N(v))$, which is the set of colors that do not appear in the neighborhood of $v$ in $X_t$; this is the heat-bath version of the Glauber dynamics.

When $k\geq\Delta+2$ then the Glauber dynamics is ergodic and hence there is a unique stationary distribution.  Moreover, since the transitions of the Glauber dynamics are symmetric then the stationary distribution is uniform over $\Omega$; our new chain is not necessarily symmetric. 
The {\em mixing time} is the number of steps, from the worst initial state $X_0$, so that the chain is within total variation distance $\leq 1/4$ of the stationary distribution of the Markov chain, see~\cref{sec:MC} for a more formal definition.

There are various attempts at running asynchronous versions of the Glauber dynamics in the distributed setting, namely HOGWILD!~\cite{Smola,ZR}, but there are few theoretical results and the resulting process is not guaranteed to have the correct asymptotic distribution~\cite{DOR16,DDJ18,TSD20}.
There is also considerable work in constructing distributed sampling algorithms, including distributed versions of the Glauber dynamics~\cite{FG18,FY18,JLY18,FSY17,FHY18,LY22}; we discuss below the relevant results in our setting of the colorings problem.  An important caveat about previous results is that they require a strong form of decay of correlations, such as the Dobrushin uniqueness condition, and our results hold in regions where Dobrushin's uniqueness condition does not hold.

In the sequential setting, a seminal work of Jerrum~\cite{Jerrum} proved $O(n\log{n})$ mixing time of the Glauber dynamics whenever $k>2\Delta$ where $\Delta$ is the maximum degree.  Vigoda~\cite{Vigoda} presented an alternative dynamics which we will refer to as the flip dynamics and proved $O(n\log{n})$ mixing time of the flip dynamics when $k>\frac{11}{6}\Delta$.   The flip dynamics is a generalization of the Glauber dynamics which ``flips'' maximal 2-colored components (clusters) in each step by interchanging the pair of colors on the chosen cluster; Vigoda's analysis chooses particular flip probabilities which depend on the size of the chosen cluster and do not flip any cluster larger than size six.  

Vigoda's result was recently improved to $k>(\frac{11}{6}-\eps_0)\Delta$
for some fixed $\eps_0\approx 10^{-5}$ by Chen, Delcourt, Moitra, Perarnau, and Postle~\cite{CDMPP19}.  This later result of $k>(\frac{11}{6}-\eps_0)\Delta$ is the best known result for general graphs.  There are various improvements (e.g., \cite{DFHV,CLMM23}), however they all require particular girth or maximum degree assumptions; the girth is the length of the shortest cycle.

In the distributed setting,
Feng, Sun and Yin~\cite{FSY17} achieved $O(\Delta\log{n})$ rounds in LOCAL model when $k>(2+\eps)\Delta$ and $O(\log{n})$ rounds when $k>(2+\sqrt{2})\Delta$.
Fischer and Ghaffari~\cite{FG18}, and independently, Feng, Hayes and Yin~\cite{FHY18}, presented a distributed version of the Glauber dynamics which converges in $O(\log{n})$ rounds for $k$-colorings on any graph of maximum degree $\Delta$ when $k>(2+\eps)\Delta$ for any $\eps>0$.  These results match Jerrum's result (in the sequential setting) for general graphs.  

We present in \cref{sec:MC} a distributed version of the flip dynamics which is a generalization of the distributed Glauber dynamics.  We match the mixing time results in the sequential setting to obtain convergence in  $O(\log{n})$ round when $k>(11/6-\eps^*)\Delta$ for some fixed $\eps^*>0$.  However, the transitions of the Markov chain we present are not symmetric and hence the stationary distribution is not the uniform distribution.

  Our proof of fast convergence of our new distributed flip dynamics utilizes the path coupling framework of Bubley and Dyer~\cite{BubleyDyer}, which is an important tool in the analysis of the mixing time for sequential Markov chains. 
  In a coupling analysis path coupling allows one to only consider ``neighboring pairs''.  In the special case of the Glauber dynamics, path coupling is related to Dobrushin's uniqueness condition but path coupling is a weaker condition (namely, Dobrushin's uniqueness condition implies path coupling). We believe our work raises the following intriguing open question.  For any spin system, or equivalently any undirected graphical model, does the path coupling condition for a local (sequential) Markov chain imply the existence of an efficient distributed algorithm which converges in $O(\log{n})$ steps?

\subsection{Motivation}

Designing a distributed algorithm for constructing a coloring is a seminal problem in the study of distributed algorithms~\cite{linial-local,NS93}.
It is an important problem in the study of symmetry breaking and is useful in the design of networking algorithms~\cite{BEPS2016, SW2010, Kuhn09,linial-local}.
One of the fundamental problems in this context that has received significant attention is minimizing the number of rounds required to construct a $(\Delta+1)$-coloring in the LOCAL model; see Barenboim, Elkin, and Goldenberg~\cite{BEG} for a recent breakthrough, and see~\cite{FYZ23,FF23} for more recent follow-up works.

Graphical models are a fundamental tool in machine learning~\cite{Murphy}, and the associated sampling problem is important for associated learning, inference, and testing problems. A noteworthy example in the history of graphical models and in the importance of the associated sampling problem is the work on Restricted Boltzmann Machines (RBMs) of Hinton~\cite{Hinton}.  An RBM is an instance of the Ising model on a bipartite graph.  The Ising model is a simpler variant of the random colorings problem in which we are sampling labellings of the vertices of a bipartite graph with only 2 colors where the labellings are weighted exponentially by the number of monochromatic edges; the generalization to $k>2$ colors is the Potts model, and the zero-temperature (antiferromagnetic) Potts model is the random colorings problem that we study.   The design of fast learning algorithms for RBMs was fundamental in the development of deep learning algorithms~\cite{Deep5,Deep4,Deep3,Deep1,Deep2}.  

Given the proliferation of machine learning tasks on high-dimensional data, there is a clear need for distributed sampling algorithms for graphical models.  For example, speeding up inference in \emph{latent Dirichlet allocation} models via parallel and distributed Gibbs sampling~\cite{ldagpu, gibbsparallel} and via the stochastic gradient sampler~\cite{stochasticlda} has received attention in the machine learning community, as has the distributed problem of finding a $k$-coloring as a subroutine for Gibbs sampling~\cite{chromatic-gibbs}.

Sampling colorings is a natural combinatorial problem to address particularly because of its importance in the study of sequential sampling algorithms.  Jerrum's sampling algorithm~\cite{Jerrum} for $k>2\Delta$ colors was a seminal work as it pioneered the use of the coupling method for sampling problems on graphical models.  As mentioned earlier, Vigoda~\cite{Vigoda} improved Jerrum's result to $k>11\Delta/6$ and this was the state of the art until the recent improvement to $k>(11/6 -\eps)\Delta$~\cite{CDMPP19}.  One of the major open problems in the area of sequential sampling is to obtain an efficient sampling scheme when $k>\Delta+1$, see~\cite{CLMM23} for the most recent progress.

Our general question is whether efficient sequential sampling schemes yield efficient distributed sampling algorithms, by which we mean an $O(\log{n})$ round algorithm in the LOCAL model.  A distributed version of the Metropolis version of the Glauber dynamics for colorings was introduced in~\cite{FG18,FHY18} and was proved to be an efficient distributed sampling scheme when $k>(2+\eps)\Delta$ for all $\eps>0$.  
Our work goes beyond the single-site Glauber dynamics to designing efficient distributed sampling schemes for more general dynamics.
 
\subsection{Technical Contribution}
\label{sec:technical-contribution}

 Recall that the Glauber dynamics updates a single vertex in each step. Several recent works present and analyze distributed versions of the Glauber dynamics (specifically, the Metropolis version) in various contexts~\cite{FG18,FHY18,LY22,FSY17}.  For more general MCMC algorithms which update larger regions than a vertex in each step, do efficient convergence results in the sequential setting for such Markov chains yield efficient distributed sampling algorithms?

A prime example to consider for this more general question is Vigoda's flip dynamics~\cite{Vigoda}. Attaining a distributed version of the flip dynamics is more challenging as we need to simultaneously recolor clusters of up to 6 vertices; here a cluster refers to a maximal 2-colored component and the recoloring acts by interchanging the respective pair of colors on each cluster. Our first contribution is presenting a distributed version of Vigoda's flip dynamics.  The challenge is to make a distributed version which is efficient but simple enough that we can still analyze it.  

To parallelize the cluster recolorings, we need to ensure that no two overlapping clusters are simultaneously active, and that no two neighboring clusters that share colors are both active.  On the other hand, we need to ``activate'' each cluster for potential recoloring with a sufficiently large probability to obtain a mixing time that is independent of the maximum degree, namely $O(\log{n})$.

Our analysis of our distributed version of Vigoda's flip dynamics follows the high-level coupling presented in Vigoda's original work~\cite{Vigoda}.  A coupling analysis of a Markov chain, considers two copies of the Markov chain (in this case the distributed flip dynamics), each with arbitrary starting states.  Our aim is that there are ``coupled transitions'' for the two chains so that after $O(\log{n})$ steps the two chains have coalesced in the same state with sufficiently large probability; by coupled transition we mean that the two chains can couple their transitions as long as when viewed in isolation, each is a faithful copy of the original Markov chain.  The idea is that if we consider one of the chains to be in the stationary distribution, then we showed that after $O(\log{n})$ steps our algorithm has likely reached the stationary distribution and hence the mixing time is $O(\log{n})$.

There are several important technical challenges that arise when doing a coupling analysis in the distributed setting for the flip dynamics.  First, we need to ensure that the clusters we flip (which means swap the pair of colors in a maximal 2-colored component) do not interfere with any other clusters we might flip by either overlapping, or by neighboring and containing a common color.  Subsequently when we do try to couple a pair of flips in the two coupled chains, we need to consider the case that one of these two clusters is not flippable in only one chain due to one of these aforementioned conflicts (such as an overlapping cluster in only one of the chains). 

Finally, we use the path coupling framework~\cite{BubleyDyer} which allows us to restrict attention to the case that the pair of coupled chains only differ at a single vertex $v^*$; this was crucial in Vigoda's original analysis as well.  Vigoda's analysis only needed to consider clusters that include~$v^*$ or that neighbor this disagree vertex~$v^*$.  In our setting we also need to analyze and couple clusters that are distance 2 away from~$v^*$, where distance is measured by cluster flips; this is due to a neighboring cluster possibly being flippable in only one of the chains and then this effect reverbates out.  Moreover, those coupling for clusters at distance $<2$ from~$v^*$ is more complicated than the sequential setting as a coupled cluster might not be flippable in only one chain due to conflicts with other clusters.

\section{Distributed Flip Dynamics}
\label{sec:dist-defn}

Here we introduce the necessary graph theory and Markov chain definitions and then present the distributed version of the flip dynamics which we analyze in \cref{sec:coupling}.

\subsection{Preliminaries}\label{sec:prelim}

Let $[k]=\{1,\dots,k\}$.  
For a graph $G=(V,E)$ let $i\sim j$ denote $(i,j)\in E$, and for $v\in V$, let $N(v)=\{w\in V: v\sim w\}$ denote the neighbors of a vertex $v$.  For integer $k\geq 2$, let $\Omega^* = [k]^V$ denote the set of $k$-\emph{labellings} and $\Omega = \{\sigma\in [k]^V: \mbox{for all } i\sim j, \sigma(i)\neq\sigma(j)\}$ denote the set of $k$-\emph{colorings} of $G$.  Throughout this paper, a coloring (or $k$-coloring) refers to a proper vertex $k$-coloring.

\subsection{Clusters}

For a coloring $\sigma$, a cluster $S$ in $\sigma$ is a maximal 2-colored component of {\em size at most 6}; this is formally defined in the following definition.

\begin{definition}
Let $G=(V,E)$ be a graph and $\sigma\in\Omega^*$.  
For a vertex $v\in V$ and color $c\in [k]$, let $S_{\sigma}(v,c)$ denote the set of vertices reachable from $v$ by a $(\sigma(v),c)$ alternating path (recall, a $(c',c)$ alternating path is a sequence of vertices $v_1,v_2, \ldots, v_\ell \in V$ for some $\ell \geq 1$ such that $\sigma(v_i) \in \{c',c\}$ for all $1 \leq i \leq \ell$, and $(v_i, v_{i+1}) \in E$ and $\sigma(v_i) \neq \sigma(v_{i+1})$ for all $1 \leq i \leq \ell -1$).  Note, that even though $\sigma$ might be an improper coloring, the definition of an alternating path still requires the colors of adjacent vertices on the path to be different in $\sigma$.

When $|S_{\sigma}(v,c)|\leq 6$ then we refer to $S=S_{\sigma}(v,c)$ as a {\em cluster}.
Let \[
\SS_{\sigma}=\bigcup_{v\in V,c\in [k]} \{S_{\sigma}(v,c):|S_{\sigma}(v,c)|\leq 6\},
\]
denote the collection of all clusters in $\sigma$ of size at most 6, where the size of a cluster refers to the number of vertices in the cluster.  The restriction to size at most 6 is due to the Markov chain used as in previous works~\cite{Vigoda,CDMPP19}.
\end{definition}

For a labelling $\sigma\in\Omega^*$, vertex $v\in V$, and color $c\in [k]$, the {\em flip} of cluster $S_{\sigma}(v,c)$ interchanges colors $\sigma(v)$ and $c$ on the set $S_{\sigma}(v,c)$.   Let $\sigma'$ denote the resulting coloring after this flip of cluster $S_{\sigma}(v,c)$.  Notice that if $\sigma\in\Omega$ then $\sigma'\in\Omega$, i.e., if it is a proper coloring before the flip, then after the flip it remains a proper coloring since the clusters are maximal 2-colored components.  This is the key property for the flip dynamics as once we reach a proper coloring then we are guaranteed to stay at proper colorings.
Nevertheless, the definition of a cluster $S_\sigma(v,c)$ is also defined for improper colorings $\sigma\in\Omega^*\setminus\Omega$; this enlarged state space of improper colorings is used in the proof but not in the algorithm itself, see~\cref{sec:proof-weaker} for further discussion of this technicality.

Consider a coloring $\sigma\in\Omega$ and a vertex $v\in V$.  For every color $c$ which does not appear in the neighborhood of $v$, i.e., $c\not\in \sigma(N(v))$ then the corresponding cluster is of size 1, i.e., $|S_\sigma(v,c)|=1$ since $S_\sigma(v,c)=\{v\}$.  Flips of these singleton clusters are exactly the transitions of the Glauber dynamics.  
The flip dynamics of Vigoda~\cite{Vigoda} is a generalization of the Glauber dynamics in which clusters of size at most 6 are flipped with positive probability (depending on the size of the cluster).
Note, for $c=\sigma(v)$ then we get a singleton cluster and the flip does not change the coloring, hence the flip dynamics has a non-zero self-loop probability and thus is aperiodic.

For clusters $S,T\in \SS_{\sigma}$, we say $S$ and $T$ are neighboring clusters, which we denote as $S\sim T$, if there exists $v\in S$ and $w\in T$ where $v\sim w$.

\subsection{Markov Chains}
\label{sec:MC}
Consider a Markov chain $(X_t)$ with state space $\Omega$ and transition matrix $P$ and unique stationary distribution~$\pi$.
We say that the chain is \emph{aperiodic} if $\gcd\{t \colon P^t(x, x) > 0\} = 1$ for all $x \in \Omega$ and \emph{irreducible} if for all $x, y \in \Omega$, there exists a $t$ such that $P^t(x,y) > 0$. 
Recall that if the chain is both aperiodic and irreducible, then it is \emph{ergodic}
and the chain has a unique \emph{stationary distribution} $\pi$ where:
$ 
     \text{for all } x, y \in \Omega, \lim_{t \rightarrow \infty} P^t(x, y) = \pi(y).
$

The {\em mixing time} is the number of steps, from the worst initial state $X_0$, until the chain is within total variation distance $\leq 1/4$ of the stationary distribution:
\[
    \Tmix := \max_{x \in \Omega} \min \{t  \mid \Vert P^{t}(\sigma, \cdot) - \mu \Vert_{\rm TV} \leq 1/4\},
\]
where $d_{TV}$ is the \emph{total variation distance},
$
     \Vert \mu - \omega \Vert_{\rm TV} := \frac{1}{2} \sum_{x \in \Omega} \lvert \mu(x) - \omega(x) \rvert.
$
The choice of constant $1/4$ is somewhat arbitrary since, for any $\varepsilon>0$, we can obtain total variation distance $\leq\varepsilon$ after $\leq\log(1/\varepsilon) \Tmix$ steps.

\subsection{Algorithm Description}
\label{sec:distributed}

We begin by defining a sequential process and then show that this process can be implemented efficiently in a distributed manner.

We have the following parameters in our algorithm.  
Let $\alpha:=\eps/(5000k)$ where $k\geq (11/6+\eps)\Delta$ for some $\eps>0$.
The parameter $\alpha$ will be used for the activation probability of a cluster.  In~\cref{sec:LP} when we strengthen the main result for $k<(11/6)\Delta$ then we will 
redefine $\alpha$ so that it depends on the distance of $k$ below 
$(11/6)\Delta$.

Let $1\geq f_i\geq 0$ for all $i\geq 1$ be a sequence of ``flip'' probabilities that contain the following key properties:
$f_1=1$, $f_i\geq f_{i+1}$ for all $i$, and $f_i=0$ for all $i\geq 7$.  The following process is well-defined for any choice of flip probabilities with these properties.  To prove~\cref{thm:main-weaker,thm:main} we will choose slightly different flip probabilities. In particular, to prove the slightly weaker result (\cref{thm:main-weaker}) in~\cref{sec:weaker} we will choose flip probabilities as in~\cite{Vigoda}, and then to get the refined result (\cref{thm:main}) in~\cref{sec:LP} we will use the setting in~\cite{CDMPP19}.

We now define the Markov chain $\MCflip$ with state space $\Omega$. 
For a coloring $X_t\in\Omega$, the transitions $X_t\rightarrow X_{t+1}$ of $\MCflip$ are defined as follows:

\begin{enumerate}
\item
Independently for each $S\in\SS_{\sigma}$, cluster $S$ is active with probability $\alpha$.
\item 
A cluster $S=S_{X_t}(v,c)$ is flippable if the following hold: 
\begin{enumerate}[(a)]
\item $S$ is active;
\item {\em Overlapping clusters:} There is no active $S'\neq S$ where $S\cap S'\neq\emptyset$;
\item {\em Conflicting neighboring clusters:} For all active clusters $T=T_{X_t}(w,c')$ where $S\sim T$,
$\{X_t(v),c\}\cap\{X_t(w),c'\}=\emptyset$.  \label{step:conflict_n}
\end{enumerate}
\item Independently for each flippable cluster $S$, flip $S$ with probability $f_i$ where $i=|S|$.
\item Let $X_{t+1}$ denote the resulting coloring.
\end{enumerate}

Notice that step~\ref{step:conflict_n} is saying that for a pair of active and neighboring clusters $S$ and $T$, the pair of colors defining cluster $S$ are disjoint from the pair of colors defining cluster $T$.

The Markov chain $\MCflip$ is ergodic and hence there is a unique stationary distribution.  However, the transitions of the Markov chain are not symmetric; an earlier version of this paper incorrectly claimed that it was symmetric.  Consequently, the stationary distribution of the Markov chain is a distribution over $\Omega$ which is not necessarily the uniform distribution.

\begin{remark}
\label{rem:error}
    Here is a simple example demonstrating that the transitions of the Markov chain $\MCflip$ are not symmetric (this error was pointed out to us by Tom Hayes).
    Let $G = (V,E)$, $V = \{v_1,v_2,v_3,v_4\}$, $E = \{\{v_1,v_2\},\{v_2,v_3\}, \{v_3,v_4\}, \{v_4,v_1\}\}$, and $\Omega$ be the set of all proper colorings of $G$ with colors $\{R,B,P,Y\}$. 
 Consider the following pair of colorings $\sigma$ and $\tau$:
    let $\sigma(v_1) = R$, $\sigma(v_2) = \sigma(v_4) = B$, and $\sigma(v_3) = P$, and let $\tau(v_1) = \tau(v_3) = Y$ and $\tau(v_2) = \tau(v_4) = B$. 
    Note, $\lvert \SS_\sigma \rvert = 12$ while $\lvert \SS_\tau \rvert = 13$, and observe that $\ProbCond{X_{t+1} = \tau}{X_{t} = \sigma} \neq\ProbCond{X_{t+1} = \sigma}{X_{t} = \tau} $; hence, the chain is not symmetric for this example.
\end{remark}

The Markov chain $\MCflip$ can be implemented efficiently in a distributed manner.
\begin{lemma}
\label{lem:LOCAL}
Each step of the Markov chain $\MCflip$ can be implemented in the LOCAL model in $O(1)$ rounds.
\end{lemma}
\begin{proof}
    We describe the steps of the algorithm and how to implement them in the LOCAL model. At a given time step~$t$, denote the current coloring as~$\sigma = X_t$.
    \begin{enumerate}
        \item \label{step:findcluster} For each vertex~$v\in V$ and for each color~$c \in [k]$, identify the cluster~$S_\sigma(v, c)$. We accomplish this step by (i) sending a message indicating the index of~$v$ to each neighboring vertex~$w$ with~$\sigma(w) = c$, (ii) passing this message, along with the index of~$w$, to each neighbor~$x$ of~$w$ with~$\sigma(x) = \sigma(v)$, and (iii) repeating this process for up to six rounds. After the six rounds, each vertex has received the identities of all other vertices in its six-hop neighborhood with which it might share a cluster, and thus can determine the clusters (and their sizes) to which it belongs.  Moreover, any 2-colored components of size $>6$ will be identified and discarded.
        \item \label{step:findpres} Fix an arbitrary ordering of the vertex set of $G$ and for each cluster~$S$, identify~$\pres(S)$ to be the lowest-index vertex~$v \in S$. 
        We can accomplish this step by letting each vertex~$u \in S$ compare its own index to each of the indices of other vertices in~$S$, which have been passed during step~\ref{step:findcluster}.
        \item For each cluster~$S$, activate~$S$ with probability~$\alpha$. More precisely, for each~$v\in V$ and for each~$c\in [k]$, if~$v = \pres(S_\sigma(v, c))$, activate~$S_\sigma(v, c)$ by sending a message to every~$u\in S_\sigma(v, c)$.
        \item \label{step:conflicts} Detect conflicts: 
        \begin{enumerate}[(a)]
        \item {\em Overlapping clusters:} For all~$v \in V$, if~$S_\sigma(v, c), S_\sigma(v, c')$ are both active for some~$c \neq c'$, send messages to~$\pres(S_\sigma(v, c)), \pres(S_\sigma(v, c'))$ to ``deactivate''~$S_\sigma(v, c), S_\sigma(v, c')$.
        \item {\em Conflicting neighboring clusters:} For all~$v\in V$, for every neighbor~$w$ of~$v$, if there exist clusters~$S_\sigma(v, c) \neq S_\sigma(w, c')$ such that $\{\sigma(v), c\} \cap \{\sigma(w), c'\} \neq \emptyset$ and if~$S_\sigma(v, c)$ and~$S_\sigma(w, c')$ are both active,
        deactivate~$S_\sigma(v, c)$ and~$S_\sigma(w, c)$ (by sending messages to~$\pres(S_\sigma(v, c))$ and~$\pres(S_\sigma(w, c))$).
        \end{enumerate}
        \item For all~$v\in V$, for all~$c\in[k]$, if~$S_\sigma(v, c)$ is still active and~$v = \pres(S_\sigma(v, c))$, flip~$S_\sigma(v, c)$ with probability~$f_i,$ where~$i = |S_\sigma(v, c)|$ (by sending a message to each~$w \in S_\sigma(v, c)$ to change its color from~$c$ to~$\sigma(v)$ or vice versa).
    \end{enumerate}
    Since, in step 5, only~$\pres(S)$ is responsible for flipping~$S$, the probability of a given cluster~$S$ being flipped, conditioned on~$S$ being active and having no active neighboring or overlapping cluster, is~$f_{|S|}$. 
    
    Each of the above steps requires a constant number of rounds, proving the claim. 
    Furthermore, the amount of computation performed locally at each vertex depends only (and polynomially) on the maximum degree of the graph and the number of colors. That is, not only is the number of rounds in the LOCAL model~$O(1)$, but also the algorithm is efficient with respect to the local computation performed in each round.
\end{proof}

We present the following mixing time result:

\begin{theorem}
\label{thm:main}
There exists $\eps^*>0$, for all $\Delta\geq 2$, for any $k>(11/6-\eps^*)\Delta$ and any graph of maximum degree $\Delta$, 
 the mixing time of $\MCflip$ is $O(\log{n})$ where $n=|V|$.
\end{theorem}

Note again, the Markov chain $\MCflip$ is not necessarily symmetric, see \cref{rem:error}, and hence the stationary distribution is not the uniform distribution.

\section{Analysis of Distributed Flip Dynamics}
\label{sec:weaker}
Here we prove the following weaker version of \cref{thm:main} showing fast mixing when $k>(11/6 + \eps)\Delta$ for any $\eps>0$.
\begin{theorem}
\label{thm:main-weaker}
For all $\eps>0$, all $\Delta\geq 2$, for any $k>(11/6+\eps)\Delta$ and any graph $G=(V,E)$ of maximum degree $\Delta$, 
 the mixing time of $\MCflip$ is $O(\log{n})$.
\end{theorem}

Our specific choice of flip probabilities for this section and for the analysis in \cref{sec:coupling} are the following:
\begin{equation}
    \label{eq:flips}
 f_1:=1, \ f_2 := 13/42, \ f_3:=1/6, \ f_4:=2/21, \  f_5:=1/21, \ f_6 := 1/84, \mbox{ and } f_j:=0 \mbox{ for all } j\geq 7.
\end{equation}
These parameters match the original paper of Vigoda~\cite{Vigoda}; there are other parameter choices for which the analysis works, e.g., see~\cite{CDMPP19}, in fact, we will utilize these alternative parameters in~\cref{sec:LP}.  

We use a coupling argument that builds upon the analysis in Vigoda~\cite{Vigoda}. 
 In~\cref{sec:LP}, we further utilize the linear programming (LP) framework and the refined metric on colorings presented in Chen et al.~\cite{CDMPP19} to achieve the further improved result as stated in~\cref{thm:main}.  

\subsection{Path Coupling}
A key tool in our analysis of the mixing time is the path coupling technique of Bubley and Dyer~\cite{BubleyDyer}.
Consider an ergodic Markov chain $\MC$ with state space $\Omega$ and transition matrix~$P$.
A \emph{coupling} for $\MC$ defines, for all pairs $X_t,Y_t\in\Omega$, a joint transition
$(X_t,Y_t)\rightarrow (X_{t+1},Y_{t+1})$ such that the individual transitions $(X_t\rightarrow X_{t+1})$ and $(Y_t\rightarrow Y_{t+1})$, when viewed in isolation from each other, act according to the transition matrix $P$.
The goal is to find a coupling that minimizes the coupling time: 
$
    \Tcouple := \min\left\{t \mid \mbox{ for all }X_0,Y_0\in\Omega, \ProbCond{X_t\neq Y_t}{X_0,Y_0}\leq 1/4\right\}.
$
This implies that $\Tmix\leq\Tcouple$.

To bound the coupling time and hence the mixing time, we use the \emph{path coupling} method of Bubley and Dyer \cite{BubleyDyer} which allows us to only consider a small subset of pairs of states.  We will analyze the coupling with respect to the \emph{Hamming distance} $H(\sigma,\tau) :=\sum_{v\in V} \mathbf{1}(\sigma(v) \not = \tau(v))$. 
We present the more general form of path coupling in~\cref{sec:general-pathcoupling} which allows more general metrics.

\begin{theorem}\cite{BubleyDyer,DyerGreenhill}
\label{thm:path-coupling}
    Consider an ergodic Markov chain on $\Omega^*=[k]^V$.  Let $\beta>0$.  If for all pairs of states $X_t,Y_t\in\Omega^*$ where $H(X_t,Y_t)=1$, there exists a coupling such that:
    \[
    \ExpCond{H(X_{t+1},Y_{t+1})}{X_t,Y_t}\leq (1-\beta)H(X_t,Y_t),
    \]
    then the mixing time is bounded by
    $
        T_{mix} \leq O \left( \frac{\log(|V|)}{\beta} \right).
    $
    Moreover, the mixing time within total variation distance $\leq\delta$, for any $\delta>0$, in time $O(\log(|V|)/(\beta\delta))$.
\end{theorem}

\subsection{Overview}
We will analyze the mixing time of the chain $\MCflip$ using path coupling.
Consider a pair of colorings $X_t,Y_t$ which differ at exactly one vertex and let $v^*$ denote the disagreement, i.e., 
$X_t(v^*)\neq Y_t(v^*)$ and for all $w\neq v^*$, $X_t(w)=Y_t(w)$.
Our coupling is the identity coupling for all clusters that are the same in both chains, i.e., for all clusters $S$ where $S=S_{X_t}(w,c)=S_{Y_t}(w,c)$ for some $w\in V,c\in [k]$,
we use the identity coupling for the activation probability.  
By the identity coupling for the activation probability we mean that with probability $\alpha$ the cluster $S$ is active in both chains, and with probability $1-\alpha$ it is inactive in both chains.   
Moreover, if the cluster $S$ is flippable in both chains then we also use the identity coupling 
for the flip probability, which means that if both clusters are flippable then with probability $\alpha$ we flip the cluster in both chains and with probability $1-\alpha$ we flip the cluster in neither of the chains.  

We will define the distance of cluster $T$ from the disagree vertex $v^*$ based on the shortest path via neighboring clusters.

\begin{definition}
\label{def:distflip}
For a coloring $\sigma\in\{X_t,Y_t\}$, and a cluster $T\in\SS_{\sigma}$,  we define $\distflip(v^*,T)$ inductively as follows.  
If $T=\{v^*\}$ 
then let $\distflip_{\sigma}(v^*,T) = 0$. 
In general, let 
\[ \distflip_{\sigma}(v^*,T)=\min\{i:\mbox{there exists } S\in\SS_{\sigma} \mbox{ where } S\sim T, \distflip_{\sigma}(v^*,S)=i-1\}. 
\]
\end{definition}
See \cref{fig:clusterdist} for an illustration of  \cref{def:distflip}.
\begin{remark}
    Note, this notion of distance is equivalent to the shortest path distance from the singleton cluster $\{v^*\}$ in the {\em cluster graph}; the cluster graph is the graph on all clusters in coloring $\sigma$ where clusters $S$ and $T$ are adjacent if $S\sim T$.  Distance~$0$ clusters are the singleton sets $\{v^*\}$ for every color that does not appear in the neighborhood of $v^*$.  Distance~$1$ clusters are those that contain a neighbor of $v^*$ (regardless of whether they also contain~$v^*$).
\end{remark}

Any clusters $T$ where no vertex in $T$ is adjacent to $v^*$ are identical in the two chains, and thus, for every $i\geq 2$:
\[  T\in\SS_{X_t}, \distflip_{X_t}(v^*,T)= i \iff T\in\SS_{Y_t}, \distflip_{Y_t}(v^*,T)=i.
\] 
Similarly, the only clusters $T$ which ``disagree'' in the sense that they appear in only one chain then $T$ is at distance $1$ from $v^*$; more formally, if $T\in\SS_{X_t}\setminus\SS_{Y_t}$, then $\distflip_{X_t}(v^*,T)=1$, and if 
$T\in\SS_{Y_t}\setminus\SS_{X_t}$, then $\distflip_{Y_t}(v^*,T)=1$.
We use $\distflip(v^*,T)$ when the distances are equal, i.e., $\distflip_{X_t}(v^*,T)=\distflip_{Y_t}(v^*,T)$.

\begin{figure}
    \centering
    \includegraphics{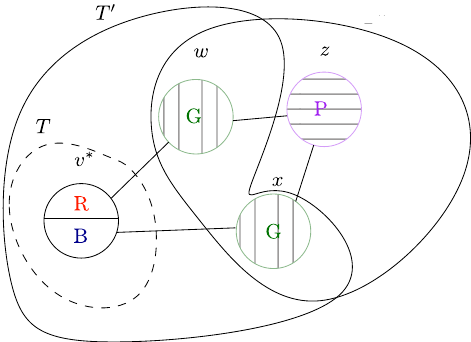}
    \caption{The vertex~$v$ has~$X_t(v) = R$ and~$Y_t(v) = B$. Here,~$T = S_\sigma(v, R) = S_\tau(v, B),$ where~$\sigma = X_t$ and~$\tau = Y_t$. By~\cref{def:distflip}, $\distflip(v^*, T) = 0, \distflip(v^*, T') = \distflip(v^*, T'') = 1$.}
    \label{fig:clusterdist}
\end{figure}

For such clusters $T$ where $\distflip(v^*,T)\geq 2$ we use the identity coupling for the activation probability in $X_t$ and $Y_t$, and thus the cluster $T$ is active in both chains or in neither chain.  It follows that for clusters~$T'$ with $\distflip(v^*,T')\geq 3$ then the cluster is flippable in both chains or in neither chain, as their neighboring active clusters are identical in the two chains.  Therefore, we can use the identity coupling for the flip probability of this cluster $T$ if the cluster is flippable, and such clusters are flipped in both chains or neither chain;
this leads to the following observation.

\begin{observation}
For any cluster $T'$ where $\distflip(v^*,T')\geq 3$,  
$
X_{t+1}(T') = Y_{t+1}(T').
$
\end{observation}

For clusters $T$ where $\dist_G(T,v^*) = 2$, it can occur that $T$ is flippable in only one of the chains (due to a neighboring cluster at distance $1$ that occurs in only one of the chains). Hence, there is a probability that such clusters can be a new disagreement. The upcoming \cref{lem:dist2-effect} proves that this occurs with an arbitrarily small constant probability.  

The following lemma bounds the expected increase in Hamming distance from flips on clusters at distance exactly $2$ from $v^*$.

\begin{lemma}
\label{lem:dist2-effect}
\[
\sum_{T: \distflip(v^*,T)=2} |T|\Prob{X_{t+1}(T)\neq Y_{t+1}(T)}
 \leq 300\Delta^2\alpha^2.
\]
\end{lemma}

We will account for these potential disagreements at distance 2 via the clusters at distance~1.  For a cluster $T$ at distance 1 to occur in only one of the chains, the pair of colors defining $T$ must include color $X_t(v^*)$ or color $Y_t(v^*)$.

\begin{proof}[Proof of \cref{lem:dist2-effect}]
Let $\SS_{X_t}\oplus\SS_{Y_t} := (\SS_{X_t}\setminus\SS_{Y_t})\cup(\SS_{Y_t}\setminus\SS_{X_t})$ denote the set of clusters that appear in one chain but not in the other chain.
Consider a cluster $S\in \SS_{X_t}\oplus\SS_{Y_t}$.
Note, all such $S$ are at $\distflip(v^*,S)=1$.  

Let $c_X=X_t(v^*)$ and $c_Y=Y_t(v^*)$.  These 
clusters $S\in \SS_{X_t}\oplus\SS_{Y_t}$ are either:
\[ 
S_{X_T}(w,c_X), S_{X_T}(w,c_Y), S_{Y_T}(w,c_X), \mbox{ or } S_{Y_T}(w,c_Y),
\] for some neighbor $w\in N(v^*)$.  Hence, there are $\leq 4\Delta$ such clusters $S\in \SS_{X_t}\oplus\SS_{Y_t}$.  

Each such cluster $S$ has size $\leq 6$ and hence it has $\leq 6\cdot 2\Delta$ neighboring clusters $T$ that share a color with~$S$.  These clusters $T$ are at distance $=2$ from~$v^*$.  Note that if $S$ and $T$ are both active then $T$ is not flippable in one of the chains, but it may be flippable in the other chain where~$S$ does not appear; hence, the chains $X_{t+1}$ and $Y_{t+1}$ potentially differ at $T$.  This yields the following:
\[
\sum_{T: \distflip(v^*,T)=2} |T|\Prob{X_{t+1}(T)\neq Y_{t+1}(T)}
\leq 6\times(4\Delta)(12\Delta)\alpha^2 = 288\Delta^2\alpha^2.
\]    
\end{proof}

We now account for the ``good moves'' where the disagreement at $v^*$ is removed.  This occurs by Glauber updates at $v^*$ where we update $v^*$ to an available color, which is a color that does
not appear in its neighborhood. 

\begin{definition}
Denote the set of {\em available colors} for $v^*$ in $X_t$ as: 
\[A(v^*)=A_{X_t}(v^*) := \{c:c\notin X_t(N(v^*))\}.
\]
\end{definition}
Note, the sets $A_{X_t}(v^*) = A_{Y_t}(v^*)$ since $v^*$ is the only disagreement at time $t$.
 Consider a color $c\in A(v^*)$.  The clusters involving~$c$ to which~$v^*$ belongs satisfy $S_{X_t}(v^*,c) = S_{Y_t}(v^*,c)= \{v^*\}$
 and hence the identity coupling is used for this cluster.  Therefore, with probability $\alpha$ the cluster is active in both chains and if 
 no active clusters overlap and no neighboring clusters have a common color then $v^*$ is recolored to $c$.

We can now bound the probability of $v^*$ agreeing at time $t+1$ in terms of the number of available colors for $v^*$. 

\begin{lemma} $\Prob{X_{t+1}(v^*)=Y_{t+1}(v^*)}\geq |A(v^*)|\alpha (1 - 10k\alpha).$
\label{lem:coupling-agree}
\end{lemma}

\begin{proof}
For each color $c\in A(v^*)$ note
$S_{X_t}(v^*,c)=S_{Y_t}(v^*,c)=\{v^*\}$.
Hence, for $c\in A(v^*)$, let $S_c=S_{X_t}(v^*,c)=S_{Y_t}(v^*,c)$ denote this cluster of size 1 which appears in both chains.  Since $S_c$ appears in both chains we use the identity coupling 
for being active so that with probability $\alpha$ the cluster $S_c$ is active in both chains, and with probability $1-\alpha$ the cluster $S_c$ is inactive in both chains.
The cluster $S_c$ may have different neighboring clusters in the two chains (which affects whether it is flippable) but if it is flippable in both chains then with probability $f_1=1$ we flip the cluster in both chains.

There are at most $2\Delta$ neighboring clusters in each chain that share a color with one of the respective cluster, and there are $k-1$ clusters (namely those at $v^*$) that overlap with these clusters. If none of the~$2\cdot 2\Delta$ neighboring clusters is active, and none of the~$2(k-1)$ overlapping clusters is active in either chain, then we can flip~$\{v^*\}$ in both chains. After this flip, $v^*$ agrees in both chains, and hence we obtain: 
\begin{align*}
 \Prob{X_{t+1}(v^*)=Y_{t+1}(v^*)} 
 & \geq |A(v^*)|\alpha(1 - \alpha)^{4 \Delta + k-1}
\\
& \geq |A(v^*)|\alpha\exp(-2\alpha(4\Delta+k-1))
\\
& \geq |A(v^*)|\alpha\exp(-10k\alpha) 
\\
&\geq |A(v^*)|\alpha(1-10k\alpha),
\end{align*}
where the second inequality uses the fact that $1-x\geq\exp(-2x)$ for $x\leq 1/2$.
\end{proof}

The upcoming lemma captures the potential disagreements that arise from flipping clusters at distance one.  
The coupling on clusters containing $v^*$ or neighboring $v^*$ in at least one chain will be coupled based on the new color $c$.
\begin{definition}
 For a color $c\in [k]$, let $N_c(v^*) = \{w\in N(v):X_t(w)=c\}=\{w\in N(v):Y_t(w)=c\}$ denote the neighbors of $v^*$ with color~$c$, and let $d_c(v^*) = |N_c(v^*)|$ denote the number of neighbors of $v^*$ with color~$c$ at time $t$.
 
 Let $\SS_{X_t}(c)$ denote the collection of clusters at distance $1$ in $X_t$ that involve color $c$:
\[ \SS_{X_t}(c):= \{S_{X_t}(w,X_t(v^*))\mid w\in N_c(v^*) \}\cup\{S_{X_t}(w,Y_t(v^*))\mid w\in N_c(v^*)\},
\]
and similarly
let $\SS_{Y_t}(c)$ denote the corresponding collection for the coloring $Y_t$.
\end{definition}

The sets $\SS_{X_t}(c)$ and $\SS_{Y_t}(c)$ are coupled with each other.  We will specify the detailed coupling later, for now all
that is needed is that these sets $\SS_{X_t}(c)$ and $\SS_{Y_t}(c)$ are coupled with each other.  
We can now state the key lemma bounding the increase in Hamming distance when we do a coupled update on these sets
$\SS_{X_t}(c), \SS_{Y_t}(c)$.

Recall, that for $\sigma,\tau\in\Omega$, $H(\sigma,\tau)=|\{v\in V:\sigma(v)\neq\tau(v)\}|$ is the Hamming distance.  For a color $c$ where $d_c(v^*)\geq 1$ we are concerned with the change in Hamming distance involved in flips with respect to color $c$ which we formalize as follows.  
For a color $c$, let 
\[ H_c(X_{t+1},Y_{t+1}) = | \{v\in V\setminus\{v^*\}:\sigma(v)\neq\tau(v), v\in S \mbox{ for some } S\in\SS_{X_t}(c)\cup\SS_{Y_t}(c)\}|.
\]
The following lemma bounds the Hamming distance, after the coupled flip, on vertices involved in $dist(v^*,S)$ $=1$ flips related to color $c$.

\begin{lemma}
\label{lem:disagree-effect}
Let $c\in [k]$ where $d_c(v^*)>0$.   
Then,
\[
\Exp{H_c(X_{t+1},Y_{t+1})}
\leq 
\alpha\left(\frac{11}{6}d_c(v^*)-1\right)(1+400k\alpha).
\]
\end{lemma}

The proof of \cref{lem:disagree-effect} is deferred to Section~\ref{sec:coupling}.  Combining the above lemmas we can prove \cref{thm:main-weaker}, which is the slightly weaker main result in that it holds for $k\geq (11/6 + \eps)\Delta$ for any $\eps>0$.

\subsection{Proof of \cref{thm:main-weaker}}
\label{sec:proof-weaker}
We can extend the definition of our Markov chain (see \cref{sec:distributed}) to be over all labellings $\Omega^*=[k]^V$ instead of just proper colorings $\Omega$;
this is necessary to apply path coupling~\cref{thm:path-coupling}.
An identical approach is used in both \cite{Vigoda} and \cite{CDMPP19}. 

The definition of the Markov chain described in~\cref{sec:distributed} is identical, we simply extend the state space.  A set $S_\sigma(v,c)$ is still defined as the set of vertices reachable from $v$ by a $(\sigma(v),c)$ alternating path.  And hence the notion of a cluster is still the same as before.  Note, that while the chain restricted to proper colorings is symmetric, this is not necessarily true for improper colorings.  
All of the bounds stated in~\cref{sec:weaker} hold for possibly improper colorings $X_t,Y_t\in\Omega^*$.

Consider a 
labelling
$X_0 \in \Omega^*\setminus \Omega$; note, $X_0$ is not a proper coloring since $X_0\not\in\Omega$. 
For $k\geq\Delta+2$, there is a sequence of transitions with non-zero probability (e.g., a sequence of Glauber moves as in the proof of irreducibility) so that it reaches a proper coloring, i.e., $X_t\in\Omega$ for some $t\geq 0$.  Moreover, for any proper coloring $X_t\in\Omega$ then it stays on proper colorings, i.e., $X_{s}\in\Omega$ for all $s\geq t$, as the process does not introduce improper colorings.  
Therefore, states in~$\Omega$ are the only ones which have positive probability in the stationary distribution, and hence the stationary distribution of the chain defined on all labellings $\Omega^*$ is the same as the stationary distribution of the original chain defined on the set of proper colorings~$\Omega$.

If the initial state is restricted to $\Omega$, i.e., $X_0$ is a proper coloring, then the chain
is identical to the process defined in~\cref{sec:MC}.
Furthermore, since the mixing time is defined from the worst initial state then a mixing time upper bound for the chain defined on $\Omega^*$ implies the same bound on the mixing time for the chain from~\cref{sec:MC} defined only on~$\Omega$.

We now have all the tools necessary to prove~\cref{thm:main-weaker}.

\begin{proof}[Proof of Theorem~\ref{thm:main-weaker}]
First consider the available colors for $v^*$.  Let $d(v^*)=\sum_{c\in [k]} d_c(v^*) = |N(v^*)|$ denote the degree of $v^*$.  Note that, since there is an extra available color for every time a color repeats in $N(v^*)$, we have the following bound on the number of available colors: 
\begin{equation} \label{eqn:available}
|A(v^*)| =  k - d(v^*) + \sum_{c\in [k]:d_c(v^*)\geq 2} (d_c(v^*)-1).
\end{equation}
Hence, we have the following bound which will be useful in the upcoming calculation:
\begin{equation}
\label{eqn:interA}
    |A(v^*)| - \sum_{c: d_c(v^*)>0} \left((11/6)d_c(v^*)-1\right) = k - d(v^*) -\sum_{c:d_c(v^*)>0}(5/6)d_c(v^*) 
    \geq k - (11/6)\Delta \geq \eps\Delta,
\end{equation}
since $k\geq ((11/6) + \eps)\Delta$.

Now by combining \cref{lem:coupling-agree,lem:dist2-effect,lem:disagree-effect} we can complete the proof of the theorem:
\begin{align}
\lefteqn{ 
\ExpCond{H(X_{t+1},Y_{t+1})}{X_t,Y_t} 
}
\nonumber
\\
\nonumber
& \qquad \leq 
1 - \Prob{X_{t+1}(v^*)=Y_{t+1}(v^*)} 
+ \sum_{c:d_c(v^*)>0}\Exp{H_c(X_{t+1},Y_{t+1})} 
\\ 
\nonumber
& \hspace*{1in} + \sum_{T:\distflip(v^*,T)= 2} 
|T|\Prob{X_{t+1}(T)\neq Y_{t+1}(T)}
\\
\nonumber
& \qquad \leq  1 - \alpha|A(v^*)|(1-10k\alpha)
+\sum_{c: d_c(v^*)>0} \left[\alpha\left(\frac{11}{6}d_c(v^*)-1\right)(1+400k\alpha)\right] 
 + \ 300\Delta^2\alpha^2
\\
\nonumber
&  
 \hspace{4.25in}
 \mbox{by \cref{lem:coupling-agree,lem:dist2-effect,lem:disagree-effect}}
\\ 
\nonumber
& \qquad \leq 
1 - \alpha\eps\Delta + 10k^2\alpha^2 + 800k\Delta\alpha^2 + 300\Delta^2\alpha^2
 \hspace{1.25in}
 \mbox{by \cref{eqn:interA}}
\\ 
\label{eqn:inter-thmweak}
& \qquad \leq 
1 - \alpha\eps\Delta + 1110k^2\alpha^2
\\
\nonumber
& \qquad \leq  1 - \eps^2/60000 
\hspace{3in}
\mbox{since $\alpha=\eps/(5000k)$,}
\end{align}
where the last inequality assumes $k\leq 3\Delta$, since the case $k>(2+\eps)\Delta$ was established by \cite{FG18,FHY18}.

Finally, applying the path coupling~\cref{thm:path-coupling} we obtain mixing time $O(\log{n})$.  Moreover, we obtain mixing time within total variation distance $\leq\delta$, for any $\delta>0$, in time $O(\log(n/\delta))$.
\end{proof}

\section{Coupling Analysis for Neighboring Clusters}
\label{sec:coupling}

We now prove \cref{lem:disagree-effect}.   We need to show that the expected change in the Hamming distance from coupled flips at distance $=1$ which involve color $c$ is at most: 
\begin{equation}
\label{eqn:4.8-main}
    \alpha\left(\frac{11}{6}d_c(v^*) -1\right)(1+400k\alpha).
\end{equation}
We note that a bound of $\frac{11}{6}d_c(v^*) -1$ was established by Vigoda~\cite{Vigoda} in the sequential setting when $k>\frac{11}{6}\Delta$.  The factor $(1+200k\alpha)$ will come from \cref{lem:flip-diff} below.

Before delving into the proof of \cref{eqn:4.8-main} we state several key properties of the settings for the flip probabilities in~\cref{eq:flips}:
\begin{enumerate}
    \item\label{Flip:P1} 
    For all integer $i,j\geq 1$, 
    $i(f_i-f_{i+1}) + (j-1)(f_j-f_{j+1}) \leq 5/6.$
    \item \label{Flip:P2}
    For all integer $i\geq 1$, 
    $2(i-1)f_i+ f_{2i+1} \leq 2/3.$
\end{enumerate}

 Fix a color $c\in [k]$ where $d_c(v^*)>0$; we will consider two cases:  $d_c(v^*)=1$ or $d_c(v^*)\geq 2$.

\subsection{Flippable Difference}
\begin{lemma}
\label{lem:flip-diff}
    For any cluster $C$,
$
    \Prob{C \mbox{ is active and not flippable}}\leq 20k\alpha^2.
    $
\end{lemma}

\begin{proof}
The cluster $C$ is active with probability $\alpha$.  Assuming $C$ is active, there are two ways that $C$ is not flippable, either (i) an overlapping cluster, or (ii) a neighboring cluster that shares a color with $C$.  For case (i), since $|C|\leq 6$ and each vertex is in $k$ clusters, then the probability of a cluster that overlaps $C$ also being active is $\leq \alpha 6k\alpha$.  For case (ii), there are $\leq 6\Delta$ neighboring vertices, each has $\leq 2$ clusters that share a color, and hence the probability of case (ii) is $\leq \alpha 12\Delta\alpha$.
Combining the above calculations we have the following:
  \[
    \Prob{C \mbox{ is active and not flippable}}\leq 
    \alpha(6k\alpha + 12\Delta\alpha)\leq 20k\alpha^2.
    \]
\end{proof}

\subsection{Color Appears Once}
\label{sec:coloronce}
Suppose $d_c(v^*)=1$. Let $w\in N(v^*)$ be the unique neighbor where $X_t(w)=Y_t(w)=c$, and let 
$R:=X_t(v^*)$ and $B:=Y_t(v^*)$.  
We are coupling the clusters in the set $\SS_{X_t}(c)$ with 
$\SS_{Y_t}(c)$, and since $d_c(v^*)=1$ these sets are the following:
\[ \SS_{X_t}(c) = \{ S_{X_t}(w,R), S_{X_t}(w,B)\} \mbox{ and } \SS_{Y_t}(c)=\{S_{Y_t}(w,R), S_{Y_t}(w,B)\}.\]
Observe
$S_{X_t}(w,R)=S_{Y_t}(w,R)\cup\{v^*\}$
and
$S_{Y_t}(w,B)=S_{X_t}(w,B)\cup\{v^*\}$.
Let $i:=|S_{Y_t}(w,R)|$ (hence, $|S_{X_t}(w,B)|=i+1$), and 
let $j:=|S_{X_t}(w,B)|$ ($|S_{Y_t}(w,B)|=j+1$).  Note, $i,j\geq 1$.

We couple the clusters in the following manner.  
With probability $\alpha$, cluster $S_{X_t}(w,R)$ is active in $X_t$ and $S_{Y_t}(w,R)$ is active in $Y_t$, while with probability $1-\alpha$ both of these clusters are inactive .  Similarly, with probability~$\alpha$  
then both: cluster $S_{X_t}(w,B)$ is active in $X_t$ and $S_{Y_t}(w,B)$ is active in $Y_t$.  

Suppose that $S_{X_t}(w,R)$ and $S_{Y_t}(w,R)$ are both flippable.  In this case we maximize the probability that we flip both clusters.  Since $f_i\geq f_{i+1}$ then with probability $f_{i+1}$ we flip both clusters $S_{X_t}(w,R)$ and $S_{Y_t}(w,R)$, assuming they were both flippable.  Similarly, with probability $f_{j+1}$ we flip both clusters $S_{X_t}(w,B)$ and $S_{Y_t}(w,B)$, assuming they were both flippable.   Note in both of these cases where we flip both $S_{X_t}(w,R)$ and $S_{Y_t}(w,R)$ or we flip both $S_{X_t}(w,B)$ and $S_{Y_t}(w,B)$, then the Hamming distance does not change as the chains only differ at $v^*$ after the coupled update.

Suppose that all 4 clusters were flippable.  (Recall an active cluster $S$ is flippable if there is no overlapping active cluster and no neighboring active cluster that shares one of the two colors with $S$.). Then after the above coupling of $S_{X_t}(w,R)$ with $S_{Y_t}(w,R)$, and $S_{X_t}(w,B)$ with $S_{Y_t}(w,B)$, there remains probability $f_j-f_{j+1}$ to flip $S_{X_t}(w,B)$, and probability $f_i-f_{i+1}$ to flip $S_{Y_t}(w,R)$.  We maximally couple these remaining flips and hence with 
probability $\min\{f_i-f_{i+1},f_j-f_{j+1}\}$ we couple the flips of clusters $S_{X_t}(w,B)$ and $S_{Y_t}(w,R)$.  Note in this case where we flip both $S_{X_t}(w,B)$ and $S_{Y_t}(w,R)$ then the Hamming distance increases by $\leq (i+j-1)$ since $S_{X_t}(w,B)\cap S_{Y_t}(w,R)\supseteq\{w\}$.  

In the above coupling, we considered all coupled flips of cluster pairs.
For each pair, one of these clusters may be flippable and the other is not (due to a neighboring or overlapping cluster also being active).  
In that case, we flip the flippable cluster by itself, and the Hamming distance increases by at most $6$ since the cluster is of size at most $6$.   
By \cref{lem:flip-diff} the probability of this occurring for a specific cluster is at most $20k\alpha^2$.
Hence, the total increase in expected distance is at most $120k\alpha^2$.  

Let us assume without loss of generality that $i\leq j$ and hence $f_i-f_{i+1}\geq f_j-f_{j+1}$.  Now we can simplify and summarize the effect of the above coupled flips that change the Hamming distance.  
Since the clusters are active with probability $\alpha$, with probability $\leq \alpha (f_{j}-f_{j+1})$ we flip $S_{X_t}(w,B)$ and $S_{Y_t}(w,R)$ 
 and then the Hamming distance increases by $\leq (i+j-1)$.
Moreover, with probability $f_i-f_{i+1}-(f_j-f_{j+1})$ we flip $S_{X_t}(w,B)$ by itself and the Hamming distance increases by $i$.
Therefore, we have the following:
\begin{align*}
\Exp{H_c(X_{t+1},Y_{t+1})} 
&\leq \Exp{H(X_{t},Y_{t})} + \alpha\left[ (i+j-1)(f_j-f_{j+1}) + i((f_i-f_{i+1})-(f_j-f_{j+1}))\right] +  120k\alpha^2
\\ &=  
1 + \alpha\left[  i(f_i-f_{i+1}) + (j-1)(f_j-f_{j+1}) \right] +  120k\alpha^2
\\ &\leq  
1 + \alpha(5/6) +  120k\alpha^2 \hspace{1in} \mbox{by Property \ref{Flip:P1},}
\end{align*}
which completes the proof of the lemma when $d_c(v^*)=1$.

\subsection{Color Appears More Than Once}
\label{sec:colormult}

The analysis of the case when the color appears more than once, i.e., $d_c(v^*)>1$, follows the same general approach as in \cref{sec:coloronce} for the case $d_c(v^*)=1$.  In particular, we use the same high-level coupling as used by Vigoda~\cite{Vigoda} but in addition we use \cref{lem:flip-diff} to bound the probability that a cluster is flippable in one chain and the coupled cluster is not flippable in the other chain.  
We now detail the analysis for the case when $d_c(v^*)\geq 2$.

Suppose $d_c(v^*)>1$.  Let $w_1,\dots,w_d$ where $d=d_c(v^*)$ denote the neighbors of $v^*$ with color~$c$, i.e., $X_t(w_i)=Y_t(w_i)=c$ for all $1\leq i\leq d$. 
Let $c_X:=X_t(v^*)$ and $c_Y:=Y_t(v^*)$.  

\newcommand{\ahat}{\widehat{a}}
\newcommand{\ihat}{\widehat{i}}
\newcommand{\fhat}{\widehat{f}}
\newcommand{\Ahat}{\widehat{A}}
\newcommand{\shat}{\widehat{s}}
\newcommand{\that}{\widehat{t}}
\newcommand{\Shat}{\widehat{S}}
\newcommand{\That}{\widehat{T}}

In $X_t$, the clusters we are considering are $S_{X_t}(w_i,c_Y)$ for $1\leq i\leq d$ and $S_{X_t}(v^*,c)$; whereas in $Y_t$, we are considering $S_{Y_t}(w_i,c_X)$ for $1\leq i\leq d$ and $S_{Y_t}(v^*,c)$.
Denote these clusters and their sizes as follows:
\begin{align}
\label{cluster1}
\mbox{$(c,c_Y)$ clusters: } &  S_i=S_{X_t}(w_i,c_Y), \ \ a_i=|S_i|, \mbox{ and } \Shat=S_{Y_t}(v^*,c), \ \ A=|\Shat|; \\
\label{cluster2}
\mbox{$(c,c_X)$ clusters: } &  T_i=S_{Y_t}(w_i,c_X), \ \ b_i=|T_i|, \mbox{ and } \That=S_{X_t}(v^*,c), \ \ B=|\That|.
\end{align}
Finally, let $a_{\max}=\max_i a_i$ maximum size of these $(c,c_Y)$ clusters and let $i_{\max}$ be the lowest index $i$ where $a_i=a_{\max}$.  Similarly, let $b_{\max}=\max_j b_j$ and $j_{\max}$ denote the index $j$ for the first neighbor achieving the max $(c,c_X)$ cluster size.

The collection of clusters considered above might not be distinct.
In particular, when $d_c(v^*)>1$ then we might have $S_i=S_j$ (or similarly $T_i=T_j$) for some $i\neq j$, due to a $(c,c_Y)$ path between $w_i$ and $w_j$ in $X_t$.   In this case, assuming $i<j$, we keep $S_i$ as it is currently defined and we set $S_j=\emptyset$ with $a_j=0$; this avoids any double-counting in the coupling definition and analysis. 
Consequently, we have $A=1+\sum_i a_i$ and $B=1+\sum_j b_j$.  (The same issue arises in previous works and is dealt with in the same manner, see~\cite{Vigoda} and~\cite[Remark A.1]{CDMPP19}.)

Our coupling is the following:
\begin{enumerate}
    \item\label{s:big} With probability $f_{A}$, couple the flips of clusters $S_{i_{\max}}=S_{X_t}(w_{i_{\max}},c_Y)$ with $\Shat=S_{Y_t}(v^*,c)$.
    \item\label{t:big} With probability $f_{B}$, couple the flips of clusters $\That=S_{X_t}(v^*,c)$ with $T_{j_{\max}}=S_{Y_t}(w_{j_{\max}},c_X)$.
    \item Let $q_{i_{\max}} = f_{a_{i_{\max}}} - f_{A}$, and for all $i\neq i_{\max}$, let $q_i=f_{a_i}$.  These are the remaining flip probabilities for the $(c,B)$ clusters which might still have flip probabilities remaining.
      \item Let $q'_{j_{\max}} = f_{b_{j_{\max}}} - f_{B}$, and for all $j\neq j_{\max}$, let $q'_j=f_{b_j}$.
    \item For $1\leq i \leq d$, with probability $\min\{q_i,q'_i\}$, couple the flips of clusters $S_i=S_{X_t}(w_i,c_Y)$ with $T_i=S_{Y_t}(w_i,c_X)$.
    \item Couple the remaining flips independently.  In particular, for $1\le i\le d$, with probability $q_i-\min\{q_i,q'_i\}$, flip $S_i=S_{X_t}(w_i,c_Y)$ and with probability 
    $q'_i-\min\{q_i,q'_i\}$, flip $T_i=S_{Y_t}(w_i,c_X)$.
\end{enumerate}

 From case 1 we have that the Hamming distance increases by $A-a_{i_{\max}}-1$ with probability~$f_A$.  Similarly, from case 2 the Hamming distance increases by $B-b_{j_{\max}}-1$ with probability~$f_B$.  From case 5 the Hamming distance increases by $a_i+b_i-1$ with probability $\min\{q_i,q'_i\}$.  Finally, from case 6 the Hamming distance increases by $a_i$ with probability $q_i-\min\{q_i,q'_i\}$ and by $b_i$ with probability $q'_i-\min\{q_i,q'_i\}$.
Putting this all together, let $\Phi$ be defined as follows:
\begin{multline*}
\Phi   :=  (A-a_{i_{\max}}-1)f_{A} + (B-b_{j_{\max}}-1)f_{B} \\
 + \sum_{i=1}^{d} \left[(a_i + b_i - 1)\min\{q_i,q'_i\} + a_i(q_i -  \min\{q_i,q'_i\}) + b_i(q'_i-\min\{q_i,q'_i\})\right].
\end{multline*}
Then we have that the expected Hamming distance for updates involving color $c\in [k]$ where $d_c(v^*)\geq 1$:
\begin{equation}
    \Exp{H_c(X_{t+1},Y_{t+1})}
\leq 
1 + \alpha\Phi + 6(4+2d_c(v^*))20k\alpha^2,
\end{equation}
where the $20k\alpha^2$ term comes from \cref{lem:flip-diff} and the fact that there are at most 4 coupled pairs of flips for cases~\cref{s:big,t:big} and at most 2 clusters for each neighbor with color $c$.

We previously completed the analysis for the case $d_c(v^*)=1$.  We divide the remaining analysis based on whether $d_c(v^*)=2$ or $d_c(v^*)>2$.

\subsubsection{Color Appears Twice}
Consider the case $d_c(v^*)=2$.
Vigoda~\cite[Claim 6]{Vigoda} proved that when $d_c(v^*)=2$, then $\Phi$ is maximized for $b_1=b_2=1$ and $a_1=a_2=a\leq 3$.  We then have: \[ A=2a+1, B=3, i_{\max}=1, j_{\max}=2, q_1=f_a-f_{2a+1}, q_2=f_a, q'_1=f_1-f_3, q'_2=f_1.
\]
Since $f_1-f_3\geq f_a-f_{2a+1}$ for $1\leq a\leq 3$ then $q_i=\min\{q_i,q'_i\}$ for $1,2$.
Hence we have the following:
\begin{align*} 
\Phi & \leq (a+1)f_{2a+1} + f_3 + a(f_a-f_{2a+1}) + ((f_1-f_3)-(f_a-f_{2a+1})) + af_a + (f_1-f_a)
    \\ &= 2 + 2(a-1)f_a + f_{2a+1}
    \\ & \leq 2+2/3,
\end{align*}
where the second line follows because $f_1 = 1$ and the third line follows by Property \cref{Flip:P2}.
Therefore, for $d_c(v^*)=2$ we have that:
\[
    \Exp{H_c(X_{t+1},Y_{t+1})}
\leq 
(8/3)\alpha + 48\times 20k\alpha^2
\leq \alpha(8/3)(1+400k\alpha),
\]
which establishes the bound in~\cref{lem:disagree-effect} for the case $d_c(v^*)=2$.

\subsubsection{Color Appears More Than Twice}
By considering the two cases  $i_{\max}=j_{\max}$ and $i_{\max}\neq j_{\max}$, one can show, as in~\cite[Proof of Lemma 5, Case (iii)]{Vigoda}, that:
\begin{equation}
    \label{Phi:333}
\Phi \leq (A-2a_{\max})f_A + (B-2b_{\max})f_B
 + \sum_{i=1}^d \left(a_if_{a_i} + b_if_{b_i} - \min\{f_{a_i},f_{b_i}\}\right).
\end{equation}
For the chosen setting of $f_i$'s, we have that $(i-c)f_i<1/4$ for all $i\geq 1$ and all $c\geq 2$.  Hence, $A-2a_{\max}<1/4$ and
$B-2b_{\max}<1/4$. For each term in the summation, assuming without loss of generality $a_i\leq b_i$ and hence $f_{a_i}\geq f_{b_i}$ then:
\[ 
a_if_{a_i} + b_if_{b_i} - \min\{f_{a_i},f_{b_i}\}
= 
a_if_{a_i} + (b_i-1)f_{b_i}
\leq f_1 + 2f_3 = 4/3,
\]
where the inequality is by observation for our choice of $f_i$'s.
Plugging these bounds back into~\cref{Phi:333} we have:
\[
\Phi \leq \frac{1}{2} + \frac{4}{3}d_c(v^*) 
\leq \frac{11}{6}d_c(v^*)-1 \ \ \mbox{ when } d_c(v^*)\geq 3.
\]

Therefore, for $d_c(v^*)\geq 3$ we have that:
\begin{align*}
    \Exp{H_c(X_{t+1},Y_{t+1})}
& \leq 
 \alpha\left(\frac{11}{6}d_c(v^*) - 1\right)
+  6(4+2d_c(v^*))20k\alpha^2 
\\
& \leq 
 \alpha\left(\frac{11}{6}d_c(v^*)-1\right) + 240d_c(v^*)k\alpha^2 + 480k\alpha^2 
\\
& \leq 
\alpha\left(\frac{11}{6}d_c(v^*)-1\right)(1+400k\alpha),
\end{align*}
which matches the bound claimed in~\cref{lem:disagree-effect} for this case $d_c(v^*)\geq 3$.
This completes the proof of~\cref{lem:disagree-effect}.

\section{Proof of Fast Mixing below 11/6}
\label{sec:LP}

\cref{sec:weaker,sec:coupling} present the proof of \cref{thm:main-weaker} which establishes $O(\log{n})$ mixing time of the distributed flip dynamics when $k>(11/6+\eps)\Delta$ for all $\eps>0$.  The improved result for $k>(11/6-\eps^*)\Delta$ for a fixed $\eps^*>0$ as stated in \cref{thm:main} is proved in this section.
Chen, Delcourt, Moitra, Perarnau, and Postle~\cite{CDMPP19} proved rapid mixing of the flip dynamics in the sequential setting for this threshold. 
Our proof combines the analysis in~\cref{sec:weaker} with the LP approach of~\cite{CDMPP19}. 

The proof of \cref{thm:main} uses the new metric introduced in~\cite{CDMPP19}, which is a weighted Hamming distance.  In particular, in~\cite{CDMPP19} they identify the configurations on the local neighborhood of the disagree vertex $v*$ for which the coupling analysis is tight, these are referred to as extremal configurations.  Hence, for a pair of configurations $X_t,Y_t$ which differ at a single vertex $v^*$, let $\gamma$ denote the fraction of neighbors of $v^*$ in non-extremal configurations.  Then, \cite{CDMPP19} defines a new weighted Hamming distance as $\wHamming(X_t,Y_t) = 1 -\gamma\eta$ for an appropriately defined small constant $\eta>0$.  

Using this new weighting, \cite{CDMPP19} proves rapid mixing of the flip dynamics in the sequential setting for $k>(11/6-\eps^*)\Delta$.  The challenge in their analysis is that one has to consider the effect of coupled flips which do not change the Hamming distance but simply change whether some neighbors of $v^*$ are in extremal configurations.

To obtain \cref{thm:main} we combine the approaches of \cite{CDMPP19} with our analysis in \cref{sec:weaker,sec:coupling} of the effect of the distributed synchronization.  However the analysis becomes considerably more complicated than in \cite{CDMPP19} because multiple clusters in the neighborhood of $v^*$ can flip in a single step, this leads to many new cases where the new weighted Hamming distance can change.  We now present the detailed analysis.

We use the following flip probabilities from \cite[Observation 5.1]{CDMPP19}:
\begin{equation}
    p_1:=1, \ p_2:=\frac{185}{616}, \ p_3:=\frac{1}{6}, \ p_4:=\frac{47}{462}, 
p_5:=\frac{9}{154}, \ p_6:=\frac{2}{77}, \mbox{ and } p_j:=0 \mbox{ for all } j\geq 7.
\end{equation}
To avoid confusion we use $p_i$ when referring to this new setting and $f_i$ for the previous setting.

Let $\eps^* = 10^{-7}$; our final result will hold for any $k>(11/6 - \eps^*)\Delta$.
Our analysis is not tight and should hold for slightly larger $\eps^*$.

\subsection{General Path Coupling: Beyond Hamming}
\label{sec:general-pathcoupling}

Chen et al.~\cite{CDMPP19} improve upon Vigoda's 11/6 by changing the metric on $\Omega$ from Hamming distance to a weighted version of Hamming distance, where the weight depends on the configuration around the disagreement. We begin by introducing the more general form of path coupling which is needed for these purposes.

\begin{definition}\label{def:pre-metric}
A pre-metric on $\Omega$ is specified by a pair $(\Gamma, \wHamming)$ where $\Gamma$ is a connected graph $(\Omega, E_\Gamma)$ and $\wHamming$ is a function $\wHamming:E_\Gamma \rightarrow \R_{> 0}$.  Then, for all $(\sigma, \tau) \in \Omega^2$, let $\wHamming(\sigma, \tau)$ be the minimum weight among all paths from $\sigma$ to $\tau$ in $\Gamma$. 
\end{definition} 

The fact that the graph $\Gamma=(\Omega,E_{\Gamma})$ is connected is important. For all $(\sigma,\tau)\in E_{\Gamma}$ we will define and analyze the coupling. However, we will analyze our coupling with respect to a new metric.  We will define this metric by specifying the distance $\wHamming(\sigma,\tau)$ for all  $(\sigma,\tau)\in E_{\Gamma}$, and then we will extend $\wHamming$ to all pairs in $\Omega^2$ by using the shortest path distance in the weighted graph $\Gamma$; this is formulated in \cref{def:pre-metric}.

We will once again need to relax the state space from $\Omega$, the set of $k$-colorings, to $\Omega^*=[k]^V$, the set of $k$-labellings, as in \cref{sec:proof-weaker}.  We will define $E_\Gamma$ to be those pairs $\sigma,\tau\in\Omega^*$ that differ at exactly one vertex; this is the same subset of pairs considered in~\cref{thm:path-coupling}.
In the simple case where 
$\wHamming=1$ for all $(\sigma,\tau)\in E_\Gamma$, then the metric $\wHamming$ produced by the pre-metric $(\Gamma,\wHamming)$ is Hamming distance $\wHamming(\sigma,\tau) = H(\sigma,\tau)$.  We will consider a more complicated distance for pairs in $E_\Gamma$, as is done in~\cite{CDMPP19}.

We will use the following more general form of the path coupling theorem.

\begin{theorem}\cite{BubleyDyer,DyerGreenhill}
\label{thm:premetcoupling}
Let $(\Gamma, \wHamming)$ be a pre-metric on $\Omega$ where $\wHamming$ takes values in $[c,1]$ for some constant $c > 0$, and $\Gamma=(\Omega,E_\Gamma)$.  Extend~$\wHamming$ to the metric on $\Omega$ defined by~\cref{def:pre-metric}.

If there exists $\beta>0$ and, for all $(X_t,Y_t)\in E_\Gamma$ there exists a coupling so that:
\[ 
\Exp{\wHamming(X_{t+1},Y_{t+1})} \leq (1-\beta)\wHamming(X_t,Y_t),
\]
then the mixing time is bounded by
\[   
T_{mix} \leq O \left( \frac{\log(|V|)}{\beta} \right).
\]  
\end{theorem}

\subsection{Defining the New Metric} 
Consider a pair of colorings~$X_t, Y_t\in\Omega$ that differ at a single vertex~$v^*$. Fix a color~$c$. Denote the neighbors of $v^*$ with color $c$ by $w_1,\dots,w_d$, where $d=d_c(v^*)$. 
Let $c_X=X_t(v^*)$ and $c_Y=Y_t(v^*)$.

Recall from~\cref{cluster1,cluster2}, the set of clusters we are considering and the notation for their sizes, which we repeat for convenience:
\begin{align*}
\mbox{$(c,c_Y)$ clusters: } &  S_i=S_{X_t}(w_i,c_Y), \ \ a_i=|S_i|, \mbox{ and } \Shat=S_{Y_t}(v^*,c), \ \ A=|\Shat|; \\
\mbox{$(c,c_X)$ clusters: } &  T_i=S_{Y_t}(w_i,c_X), \ \ b_i=|T_i|, \mbox{ and } \That=S_{X_t}(v^*,c), \ \ B=|\That|.
\end{align*}

Let a \emph{configuration} be a tuple of the form $(A, B; a, b),$
where $a = \left(a_1, \dots, a_{d_c(v^*)}\right)$ and $b = \left(b_1, \dots, b_{d_c(v^*)}\right)$.
We refer to $d=d_c(v^*)$ as the \emph{size} of a configuration~$(A, B, a, b)$.
See Figure~\ref{fig:config} for some examples. 
In general, let $(A,B; a, b) = (A^c, B^c; a^c, b^c)$ be the configuration to which a given color~$c$ belongs; we include the superscript $c$ when necessary for clarity.

\begin{figure}
    \centering
    \includegraphics[height=17em]{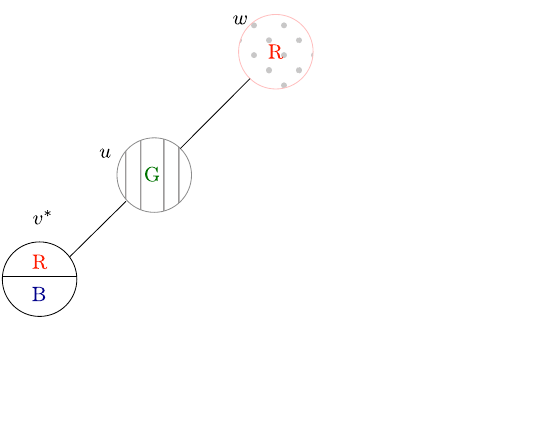}
    \includegraphics[height=17em]{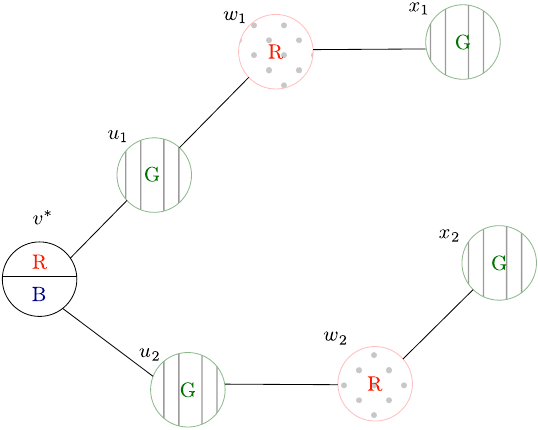}
    \caption{Examples of configuration $(3, 2; (2), (1))$ (on the left) and configuration $(7, 3; (3, 3); (1, 1))$ (on the right).}
    \label{fig:config}
\end{figure}

In Chen et al.~\cite{CDMPP19}, to prove their improved result, they showed, via a linear program (LP), that although the $k>11/6$ result is tight for the flip probabilities $\{f_i\}$ in Vigoda's original paper~\cite{Vigoda} (which are the same probabilities we used in~\cref{sec:weaker}), the analysis can be improved by weighting the single disagreement based on the configuration in the neighborhood of $v^*$.  In particular, for each color $c$ which appears in the neighborhood of $v^*$, we consider the clusters involving $c$, namely the sets $\SS_{X_t}(c)$ and $\SS_{Y_t}(c)$.  The only relevant information for these sets of clusters are the sizes of the clusters.  

In Vigoda's analysis there are six configurations of $\SS_{X_t}(c)$ and $\SS_{Y_t}(c)$ for which the analysis is tight; these are referred to as extremal configurations.  By choosing slightly different flip probabilities, Chen et al.~\cite{CDMPP19} reduce this to only two extremal configurations.  Note that all non-extremal configurations have some slack in Vigoda's analysis.
Moreover,~\cite{CDMPP19} showed that the extremal configurations are ``brittle'' and likely to flip to non-extremal configurations. 

For pairs $(\sigma,\tau)\in\Omega^2$ which differ at a single vertex $v^*$, whereas their Hamming distance is one,~\cite{CDMPP19} defined a new distance which is reduced by a factor $\eta$
for the fraction of colors that are in non-extremal configurations.  
\cite{CDMPP19} identified two specific configurations ( $(3, 2; (2), (1))$ and $(7, 3; (3, 3), (1, 1))$), which are the extremal configurations up to symmetry for their setting of flip probabilities.  

Consider a pair $\sigma,\tau\in\Omega$ which differ at a single vertex $v^*$. 
We denote the colors with these extremal configurations on the neighborhood of $v^*$ as follows:
\begin{align*}
C^1= C^1(v^*) &:= \{c: (A^c,B^c; a^c,b^c) \in\{ (3,2; (2), (1)),(2,3;(1),(2))\}\}
\\
C^2= C^2(v^*) &:= \{c: (A^c,B^c; a^c,b^c) \in\{ (7,3; (3,3), (1,1)),(3,7;(1,1),(3,3))\}\}.
\end{align*}
Denote the set of colors that do not appear in an extremal configuration with respect to $v^*$ as:
\[ C^0= C^0(v^*) := [k]\setminus (C^1\cup C^2).
\]

We denote the fraction of neighbors of $v^*$ that are in extremal configurations as:
\[\gamma(v^*)=\gamma_{\sigma,\tau}(v^*) := \left(|C^1| + 2|C^2|\right)/\Delta.
\]
Then, following \cite{CDMPP19}, we can define the following distance between this pair of states $\sigma,\tau$ which differ at a single vertex $v^*$:
\begin{equation}
\label{eq:weight}
\wHamming(\sigma, \tau) := 1 - \eta(1 - \gamma_{\sigma,\tau}(v^*)) \nonumber
\end{equation}
where~$1/2>\eta>0$ is a constant  
that we will define later. 

Note that~$\wHamming(\sigma, \tau)$ is defined when~$\sigma, \tau$ differ at a single vertex~$v^*$, and this is extended to a metric on all of $\Omega^2$ using~\cref{def:pre-metric}. 

The advantage of this refined metric is that in the case of the extremal configurations, a single move of the flip dynamics is likely to change extremal to non-extremal configurations, which reduces the distance by a factor $\eta$ even though the Hamming distance stays the same.  This yields some slack in the analysis of extremal configurations, which enables us to improve slightly on the 11/6 threshold, as is done in~\cite{CDMPP19}.

\subsection{Contraction of the New Metric}
We use the Path Coupling Theorem (\cref{thm:premetcoupling}) to prove~\cref{thm:main}. To this end, we need to show that the new metric is contracting.

Chen et al.~\cite{CDMPP19} showed that if we ignore the two extremal configurations then the maximum increase in Hamming distance per neighbor decreases from $\frac{11}{6}$ to $\frac{161}{88} = \frac{11}{6} - \frac{1}{264}$.  This is obtained by solving an appropriate LP. Using this insight, we prove the following key result.

\begin{lemma}
\label{lem:whamcontractlocal}
There exists a constant $\eps^*>0$ such that
for all $k \geq (\frac{11}{6} - \eps^*)\Delta$ the following holds. Consider~$X_t, Y_t\in\Omega^*$ which differ at a single vertex, $v^*$. 
Then
\[
\ExpCond{\wHamming(X_{t+1}, Y_{t+1})}{X_t, Y_t} \leq (1 - \beta) \wHamming(X_t, Y_t)
\]
where~$\beta > 0$ is a small constant.
\end{lemma}

\cref{thm:main} follows immediately from \cref{lem:whamcontractlocal}.

\begin{proof}[Proof of \cref{thm:main}]
If $k > 2\Delta$ then the claim follows from \cref{thm:main-weaker}.
Otherwise, combining \cref{lem:whamcontractlocal} with the path coupling~\cref{thm:premetcoupling}, we have 
\[
\Tmix \leq O\left(\frac{\log(|V|)}{\beta}\right)=O(\log{n}).
\]
This completes the proof of~\cref{thm:main}.
\end{proof}

\subsection{Contraction of the Hamming Metric}
We first show a bound on the contraction of our original Hamming metric with our flip probabilities in the distributed setting.

Let $\delta := \frac{11}{6} - \frac{161}{88}$.
We prove the following analogue of Corollary~5.1 in~\cite{CDMPP19}, which bounds the expected change in Hamming distance from coupled flips.  In the proof of \cref{lem:disagree-effect} we used that the expected change in the Hamming distance with respect to a specific color is $\leq 11/6$ for every neighbor.  In~\cite{CDMPP19} they prove that in all non-extremal neighborhoods the bound $11/6$ can be improved to $161/88=(11/6)-\delta$, while in extremal neighborhoods the bound remains at $11/6$.

The following lemma is the improved version of \cref{lem:disagree-effect}.

\begin{lemma}
\label{lem-improved:disagree-effect}
Let $\delta := \frac{11}{6} - \frac{161}{88}$.
\begin{multline}
\sum_{c\in [k]:\atop d_c(v^*)\geq 1}\Exp{H_c(X_{t+1},Y_{t+1})}
\leq 
\sum_{c \in C^0:\atop d_c(v^*)\geq 1} \alpha\left(\left(\frac{11}{6}-\delta\right)d_c(v^*)-1\right)(1+400k\alpha) 
\\
+ \sum_{c \in C^1\cup C^2} \alpha\left(\frac{11}{6}d_c(v^*)-1\right)(1+400k\alpha).
\end{multline}
\end{lemma}

 \begin{proof}
We restate \cref{lem:disagree-effect} for convenience:
\[
\sum_{c:d_c(v^*)>0} \Exp{H_c(X_{t+1},Y_{t+1})}
\leq 
\alpha\left(\frac{11}{6}d_c(v^*)-1\right)(1+400k\alpha).
\]
In \cite[see eq. (5.24) in Observation 5.2]{CDMPP19}
they show that $\Phi$ (which captures the expected change in the Hamming distance, from coupled flips, with respect to color $c$) is $\leq (\frac{11}{6}-\delta)=\frac{161}{88}$ for every non-extremal configuration, and is $\leq \frac{11}{6}$ for every extremal configuration. Applying this fact to the proof of \cref{lem:disagree-effect} gives
\begin{align}
   \lefteqn{
   \sum_{c\in [k]:\atop d_c(v^*)\geq 1} \Exp{H_c(X_{t+1},Y_{t+1})} }
   \\
   & \qquad = \sum_{c \in C^0} \Exp{H_c(X_{t+1},Y_{t+1})} + \sum_{c \in C^1 \cup C^2} \Exp{H_c(X_{t+1},Y_{t+1})} \nonumber \\
    &\qquad \leq 
    \sum_{c \in C^0:\atop d_c(v^*)\geq 1} \alpha\left(\left(\frac{11}{6}-\delta\right)d_c(v^*)-1\right)(1+400k\alpha) 
+ \sum_{c \in C^1\cup C^2} \alpha\left(\frac{11}{6}d_c(v^*)-1\right)(1+400k\alpha),
\nonumber
\end{align}
which proves the lemma.
\end{proof}

\subsection{Decomposing the New Metric}
To show that $\wHamming$ is contracting, we want to show that the expected change in this metric is bounded.
To this end, we decompose the expected change of $\wHamming$ into two parts: the expected change in $H$ and the expected change from our new weighting for colorings with extremal configurations.
For ease of notation, let $\sigma=X_t$ and $\tau=Y_t$ denote the states at time $t$, and let $\sigma'=X_{t+1}$ and $\tau'=Y_{t+1}$ denote the (random) states at time $t+1$.
We define
\begin{equation}
\label{eq:weighteddb}
B(\sigma, \tau) := H(\sigma, \tau) - \wHamming(\sigma, \tau)
\end{equation}
to denote the difference from Hamming distance, and note that $B$ is not a metric. Observe that~$\wHamming$ increases when either~$H$ increases or~$B$ \emph{decreases}. Thus, the idea now is to use the slack given by the non-extremal configurations (in \cref{lem-improved:disagree-effect}) to show that the expected \emph{increase} in the Hamming metric~$H(\sigma, \tau)$\textemdash whenever this increase is positive\textemdash is offset by the (sufficiently large) expected increase in~$B$, giving a net decrease in the metric~$\wHamming$. We do so by bounding the following division of the expected change:
\begin{align*}
\nabla_{\wHamming} = \nabla_{\wHamming}(\sigma, \tau) &:= \Exp{\wHamming(\sigma', \tau') - \wHamming(\sigma, \tau)},\\
\nabla_{H} = \nabla_H(\sigma, \tau) &:= \Exp{H(\sigma', \tau') - H(\sigma, \tau)} = \Exp{H(\sigma', \tau') - 1} \text{, and}\\
\nabla_B = \nabla_B(\sigma, \tau) &:= \Exp{B(\sigma', \tau') - B(\sigma, \tau))},
\end{align*}

Note that $\nabla_H$ is simply equal to the change in Hamming distance, and hence~\cref{lem-improved:disagree-effect} bounds~$\nabla_H$. Furthermore $\nabla_\wHamming = \nabla_H - \nabla_B$. Thus our aim is to bound $\nabla_H$ from above and bound $\nabla_B$ from below.

Recall that $C^1$ and $C^2$ are the colors appearing in extremal configurations,  
and $C^0$ is the set of remaining colors.
For each time step $t$ and and $i \in \{0,1,2\}$, let $C^i_{t} = C^i_{X_t,Y_t}(v^*)$ denote $C^i$ at time $t$.
Then for each color $c\in [k]$ and for all $i,j \in \{0,1,2\}$, let $\EE_{i,j}(c)$ denote the event that  $c \in C^i_{t}$ and $c\in C^j_{t+1}$, i.e., at time $t$ color $c$ is in $C^i$ and then after the coupled update color $c$ is in $C^j$.
 Moreover, let 
 \[ E_{i,j}(c) = \Exp{\indicator(\EE_{i,j}(c))}.
 \]
Then it follows that 
\begin{equation}
    \label{eq:nablaB-summing}
\nabla_B  \geq \frac{\eta}{\Delta}\sum_{c\in C^1_t\cup C^2_t} \left[E_{1,0}(c) + 2E_{2,0}(c) - E_{1,2}(c)\right] 
  - \frac{\eta}{\Delta} \sum_{c \in C^0_t} \left[E_{0,1}(c) + 2E_{0,2}(c)\right].
\end{equation}
(We ignore $E_{2,1}(c)$ since including it would only make the bound better.)

The following lemmas are the analog of \cite[Lemma 5.2]{CDMPP19}; we both simplify their case analysis and handle some subtleties that arise from the distributed setting. 
The first lemma lower bounds the number of colors moving from extremal to non-extremal; these correspond to good moves as they decrease the distance.

\begin{lemma}
\label{lem:barnablocalneg}
\begin{equation}
\label{eq:barnablocalneg}
 \frac{1}{\Delta} \sum_{c \in C} \left[E_{1,0}(c) + 2E_{2,0}(c) \right] \geq 
\alpha(k-\Delta-2)\gamma(v^*)(1 - 20k\alpha)^3.
\end{equation}
\end{lemma}

\begin{proof}
We begin by lower bounding $E_{1,0}(c)$.
Let $c \in C^1_t$.
There exists a unique $u \in N(v^*)$ where $X_t(u)=Y_t(u)=c$, and a unique $w \in N(u)\setminus\{v^*\}$
where $X_t(w) = Y_t(w) \in \{X_t(v^*),Y_t(v^*)\}$. 
If all of the following events occur then 
$c\in C^0_{t+1}$, namely:
\begin{enumerate}[(i)]
\item $X_{t+1}(w)=Y_{t+1}(w)\notin \{X_t(v^*),Y_t(v^*)\}$;
\item $X_{t+1}(v^*)=X_t(v^*), Y_{t+1}(v^*)=Y_t(v^*)$;
\item For all $u'\in N(v^*)\setminus\{u\}$, $X_{t+1}(u')\neq c, Y_{t+1}(u')\neq c$;
\item For all $w'\in N(u)\setminus\{w,v^*\},$ $X_{t+1}(w')\notin\{X_t(v^*),Y_t(v^*)\}, Y_{t+1}(w')\notin\{X_t(v^*),Y_t(v^*)\}$.
\end{enumerate}
We will now show a lower bound on the probability that all of these events (i)-(iv) occur in the transition $(X_t,Y_t)\rightarrow (X_{t+1},Y_{t+1})$.

Beginning with (i), there are at least $k - \Delta - 2$ colors that are not in $\{X_t(v^*),Y_t(v^*)\}$ and are not in the neighborhood of $w$ in both $X_t$ and $Y_t$.
For each of these colors $c'$, the probability of flipping the singleton cluster $S_{X_t}(w, c') = S_{Y_t}(w, c')$ is at least $\alpha(1 - 20k\alpha^2)^2$ by \cref{lem:flip-diff}.
Hence, the probability of $w$ changing to a color not in $\{X_t(v^*),Y_t(v^*)\}$ is at least $\alpha(1 - 20k\alpha^2)^2(k-\Delta-2)$.

For event (ii) to not occur then $v^*$ must be recolored in at least one of the chains $X_t$ or $Y_t$.  In each chain, there are $\leq k$ clusters that recolor $v^*$, and hence the probability of event (ii) not occurring is $\leq 2\alpha k$. 
Considering event (iii), the probability of recoloring a specific vertex $u'$ to color $c$ is $\leq \alpha P_1 = \alpha$ and hence the probability of changing at least one $u\in N(v^*)\setminus\{u\}$ to the specific color $c$ in at least one of the chains $X_t,Y_t$ is $\leq 2\alpha\Delta$.  
Finally, the event (iv) is similar to (iii) but there are now two colors to exclude, and hence the probability of (iv) not occurring is $\leq 4\alpha\Delta$. 

Putting these observations together we have that the probability that at least one of the events (ii), (iii), or (iv) does not occur
is $\leq \alpha(4k+2\Delta+4\Delta)< 10k\alpha$.  Thus, the probability that (ii), (iii), and (iv) all occur 
is $\geq (1-10k\alpha)$.
Hence for $c\in C^1_t$ we have that 
\[ E_{1,0}(c)\geq \alpha(k-\Delta-2)(1 - 20k\alpha)^3.
\] 

Consider the case when $c \in C^2_t$.  
There exists a unique pair $u_1,u_2$ where $X_t(u_1)=X_t(u_2)=Y_t(u_1)=Y_t(u_2)=c$, and there exists $w_1\in N(u_2)\setminus\{v^*\}$ where $X_t(w_1)=Y_t(w_1)\in\{X_t(v^*),Y_t(v^*)\}$.   
Moreover, exactly one of the following holds: there is a unique $w_2\in N(u_2)\setminus\{v^*,w_1\}$ where $X_t(w_2)=Y_t(w_2)=X_t(w_1)=Y_t(w_1)$, or there is a unique $w_2\in N(w_1)\setminus\{v^*,u_2\}$ where $X_t(w_2)=Y_t(w_2)=c$.

The analogous events that ensure $c\in C^0_{t+1}$ are the following:
\begin{enumerate}[(i')]
\item $X_{t+1}(u_1)=Y_{t+1}(u_1)\neq c$;
\item $X_{t+1}(v^*)=X_t(v^*), Y_{t+1}(v^*)=Y_t(v^*)$;
\item For all $u'\in N(v^*)\setminus\{u_1,u_2\}$, $X_{t+1}(u')\neq c, Y_{t+1}(u')\neq c$;
\item $X_{t+1}(u_2)=X_t(u_2)=Y_{t+1}(u_2)=Y_t(u_2)=c$;
\item $X_{t+1}(w_i)=X_t(w_i)=Y_{t+1}(w_i)=Y_t(w_i)$ for $i\in\{1,2\}$.
\end{enumerate}
Note that if these events (i')-(v') all occur then $u_2$ is the only neighbor of $v^*$ with color $c$ at time $t+1$ and one of the clusters $S_{X_{t+1}}(u_2,X_t(v^*))$ or $S_{Y_{t+1}}(u_2,Y_t(v^*))$ is of size $\geq 4$ as it contains $v^*,u_2,w_1,w_2$, and hence color $c$ is non-extremal at time $t+1$ which means $c\in C^0_{t+1}$ as desired.  

Event (i') is analogous to event (i) from earlier and occurs with probability $\geq \alpha(1 - 20k\alpha^2)^2(k-\Delta-2)$.  Events (ii') and (iii') are identical to (ii) and (iii) (except one additional excluded vertex in (iii')) and hence the probability of these not occurring is $\leq \alpha(2k+2\Delta)$.  Events (iv') and (v') are analogous to (ii) and hence the probability of these not occurring is $\leq 6\alpha k$ since there are three vertices $u_2,w_1,w_2$ under consideration and $\leq 2\Delta$ clusters to consider for each of these vertices.  

Combining these bounds we have for $c\in C^2_t$ that 
\[ E_{2,0}(c)\geq \alpha(k-\Delta-2)(1 - 20k\alpha)^3.
\]
This completes the proof of the lemma.
\end{proof}

The next lemma upper bounds the bad moves, these are the colors that change from non-extremal to extremal and hence increase the distance.

\begin{lemma}
\label{lem:barnablocalpos}
\[
\frac{1}{\Delta} \sum_{c \in C} \left[E_{0,1}(c) + 2E_{0,2}(c) + E_{1,2}(c) \right] \leq 17\alpha k(1-\gamma(v^*))  + \alpha ({31}/{100}) \Delta \gamma(v^*) 
\]
\end{lemma}

\begin{proof}
Fix a color $c$ and denote it as $G$.
Suppose $G \in C_t^0$.
Also denote $X_{t}(v^*)$ as $R$ and $Y_{t}(v^*)$ as~$B$.
It follows that $R \neq B$.

We begin by upper bounding $E_{0,1}(G)$.
We will show that if $G \in C_{t+1}^1$ then one of $\leq 2k+3\Delta$ clusters must flip at step $t$. 
Since each cluster flips with probability at most $\alpha p_1 = \alpha$, it follows that $E_{0,1}(c) \leq \alpha(2k + 3 \Delta) \leq 5 \alpha k$.

We start by considering the case that $v^*$ does not change ($X_t(v^*) = X_{t+1}(v^*)$ and $Y_t(v^*) = Y_{t+1}(v^*)$).
In this case, we will show that one of at most $\max\{\Delta, 2k+\Delta, 2k\} = 2k+\Delta$ clusters must flip for $G \in C^1_{t+1}$ to occur. 
We consider three cases depending on $d_G$ at step $t$:
\begin{enumerate}
    \item Suppose $d_G = 0$ at step $t$. \label{lem:barnablocalpos:case:1}
    For $G \in C^1_{t+1}$, there must be a new $G$ neighbor of $v^*$.  Each neighbor of $v^*$ has probability at most $\alpha p_1=\alpha$ of flipping to the specific color $G$, hence $E_{0,1}(c)\leq \Delta$ in this case.
        
    \item Suppose $d_G = 1$ at step $t$.
    We will show that in this case, one of at most $\Delta + \max\{2\Delta, k, 2k \} = 2k + \Delta$ clusters must flip for $G \in C^1_{t+1}$. 
        Let $u_t$ and $u_{t+1}$ be the unique $G$ neighbor of $v^*$ at step $t$ and $t+1$, respectively.
        
    Suppose $u_t\neq u_{t+1}$.
    Then $u_{t+1}$ recolored to $G$ at time $t$.
    There are at most $\Delta$ choices for $u_{t+1}$ and at most one cluster that when flipped recolors this specific $u_{t+1}$ to $G$, and hence the expected contribution to $E_{0,1}(c)$ is at most $\Delta$ in this case.
        
    Now suppose $u_t=u_{t+1}$ and denote it by $u=u_t=u_{t+1}$.
    We break this case into three subcases that depend on $u$'s neighborhood.  For $u\in N(v^*)$, integer $s\geq 0$ and color $c'$, let 
    \[ d_{c',s}(u) := |\{w\in N(u)\setminus\{v^*\}: X_s(w)=c' \mbox{ or } Y_s(w)=c'\}|.
    \]
    Note, $d_{R,t+1}(u)+d_{B,t+1}(u)=1$ since $G\in C^1_{t+1}$.  We now partition the cases based on the number of $R$ and $B$ neighbors of $u$ at time $t$.
    \begin{enumerate}
        \item Suppose $d_{R,t}(u) + d_{B,t}(u) = 0$.
        Since $d_{R,t+1}(u)+d_{B,t+1}(u)=1$, as pointed out earlier, a neighbor of $v^*$ must flip to $R$ or $B$ at step $t$ and thus $E_{0,1}(c)\leq \Delta + 2\Delta = 3 \Delta$ in this case (the $\Delta$ term comes from the case $u_t\neq u_{t+1}$).
        
        \item Suppose $d_{R,t}(u) + d_{B,t}(u) = 1$.
        Let $w_t$ and $w_{t+1}$ be the unique $R$ or $B$ neighbor of $u$ at step~$t$ and $t+1$, respectively.
    
        Suppose $w_t\neq w_{t+1}$.    Thus $w_{t+1}$ recolored to $R$ or $B$ at step~$t$.  There are at most $2\Delta$ clusters to consider and hence $E_{0,1}(c)\leq \Delta + 2\Delta = 3 \Delta$ in this case.

        Now suppose $w_t=w_{t+1}$.
        Since $G \in C^0_{t}$ and $w$ is the only $R$ or $B$ neighbor of $u$ (ignoring $v^*$), then there exists at least one neighbor $z \in N(w) \setminus\{u,v^*\}$ where $X_{t}(z) = Y_{t}(z) = G$.
        Since $G \in C^1_{t+1}$, then $z$ must recolor to a non-$G$ color. Since at least one of $\leq k$ clusters containing $z$ must flip at time $t$, hence $E_{0,1}(c)\leq \Delta +2k $ in this case.

        \item Suppose $d_{R,t}(u) + d_{B,t}(u) \geq 2$.
        For $G \in C^1_{t+1}$, a $R$ or $B$ neighbor of $u$ (that is not $v^*$) must flip to not $R$ or $B$ at step $t$.
        Fix any $w_1, w_2 \in N(u) \setminus \{v^*\}$ distinct $R$ or $B$ neighbors of $u$ at step $t$.
        It follows that $w_1$ or $w_2$ must flip to not $R$ or $B$. 
        For each $w_i \in \{w_1,w_2\}$, there are at most $k-2$ clusters that contain $w_i$ and flip it to a color that is not $R$ or $B$, and hence $E_{0,1}(c)\leq \Delta + 2k$ in this case.
    \end{enumerate}
     
    \item Suppose $d_G \geq 2$ at step $t$. \label{lem:barnablocalpos:case:3}
Consider any two $G$ neighbors of $v^*$ at time $t$, denote these as $u_1,u_2$.  Since there is exactly one $G$ neighbor of $v^*$ at time $t+1$ thus at least one of these two vertices $u_1,u_2$ must recolor to a different color at step~$t$.  Hence, at least one of $\leq 2k$ clusters need to flip to recolor one of these neighbors, and thus $E_{0,1}(c)\leq 2k$ in this case.
\end{enumerate}

Finally, we consider the case when $v^*$ changes (i.e., $X_{t+1}(v^*) \neq R$ or $Y_{t+1}(v^*) \neq B$).
For $G \in C^1_{t+1}$ then $X_{t+1}(v^*) \neq Y_{t+1}(v^*)$.
It follows that in exactly one chain, a cluster containing $v^*$ and one of its neighbors must flip.
That is, $v^*$ must flip to the color of a neighbor in one chain. 
Since $v^*$ has $\Delta$ neighbors, there are at most $\Delta$ clusters that flip $v^*$ to a color in its neighborhood in one chain,
thus, one of at most $2\Delta$ clusters must flip for $G \in C^1_{t+1}$, and hence $E_{0,1}(c)\leq 2\Delta$ in this case.

We now want to bound $E_{0,2}(c)$. 
Similar to the previous case, we will show that if $G \in C_{t+1}^2$, one of $3k+3\Delta$ clusters must flip at step $t$. 
Since each cluster flips with probability at most $\alpha p_1 = \alpha$, it follows that $E_{0,2}(c) \leq \alpha(3k + 3 \Delta) \leq 6 \alpha k$.

We start by assuming that $v^*$ does not change (i.e., $X_t(v^*) = X_{t+1}(v^*)$ and $Y_t(v^*) = Y_{t+1}(v^*)$).
In this case, we will show that one of at most $\max\{\Delta, 3k+\Delta, 3k\} = 3k+\Delta$ clusters must flip for $G \in C^2_{t+1}$. 
We consider three cases depending on $d_G$ at step $t$:
\begin{enumerate}
    \item Suppose $d_G \leq 1$ at step $t$. 
    For $G \in C^2_{t+1}$, it must be that $d_G = 2$ at step $t+1$.
    Hence, a non-$G$ neighbor of $v^*$ must flip to $G$ at step $t$ and therefore $E_{0,2}(c)+E_{1,2}(c)\leq\Delta$ in this case.
        
    \item Suppose $d_G = 2$ at step $t$.
    Let $u_1,u_2$ be the unique $G$ neighbors of $v^*$ at step $t$.
    We will show that in this case, one of at most $\Delta + \max\{2\Delta,\Delta + 2k,  2k,3k\} = 3k+\Delta$ clusters must flip for $G \in C^2_{t+1}$. 
    
    (i) Suppose $u_1$  is not $G$ or $u_2$ is not $G$ at step $t+1$.
    Then a neighbor of $v^*$ besides $u_1$ and $u_2$ must flip to $G$ at step $t$. 
    For each $u' \in N(v^*) \setminus \{u_1,u_2\}$, there is one cluster that contains $u'$ and flips it to $G$.
    Since at most $\Delta$ neighbors of $v^*$, one of at most $\Delta$ clusters must flip for $G \in C^2_{t+1}$ in this case. 
    
    (ii) Now suppose $u_1$ and $u_2$ are $G$ at step $t+1$. 
    We break this case into three subcases that depend on the clusters containing $u_1$ and $u_2$.
    Since $G \in C^0_{t}$, it follows that there is at least one $u \in \{u_1,u_2\}$ such that a vertex in its neighborhood of the second neighborhood must flip to or from a color in $\{R,B,G\}$.
    \begin{enumerate}
        \item Suppose $d_R(u) + d_B(u) = 0$ at step $t$.
        For $G \in C^2_{t+1}$, it must be that $d_R(u) + d_B(u) = 1$ at step $t+1$.
        Hence, a neighbor of this specific $u$ must flip to $R$ or $B$ at step $t$, and therefore $E_{0,2}(c)\leq \Delta + 2\Delta \leq 3\Delta$ in this case (where the $\Delta$ term is coming from case (i)).
        
        \item Suppose $d_R(u) + d_B(u) = 1$ at step $t$.
        Let $w$ be the unique $R$ or $B$ neighbor of $u$ at step $t$ that is not $v^*$.
        We will show that one of at most $\Delta + \max\{\Delta,k,2k\} = \Delta + 2k$ clusters must flip for $G^2_{t+1}$.
        
        Suppose $w$ is not $R$ or $B$ at step $t+1$.
        Then a different neighbor must flip to $R$ or $B$ at step~$t$.
        For each neighbor $w' \in N(u) \setminus \{v^*\}$, there are at most two clusters that contain $w'$ and flip it to $R$ or $B$. 
        Since there are at most $\Delta$ neighbors of $u$, one of at most $\Delta$ clusters must flip for $G \in C^2_{t+1}$ in this case.
    
        Now suppose $w$ is $R$ or $B$ at step $t+1$.
        We consider three cases for $d_G(w)$:
        \begin{enumerate}
            \item Suppose $d_G(w) = 0$ at step $t$. 
            For $G \in C^2_{t+1}$, it must be that $d_G(w) = 1$ at step $t+1$.
            Hence, a non-$G$ neighbor of $w$ must flip to $G$ at step $t$.
            There are $\leq\Delta$ neighbors of $w$.
            For each neighbor $z \in N(w)$, one cluster contains $z$ and flips it to $G$.
            Thus, in this case, one of at most $\Delta$ clusters must flip for $G \in C^2_{t+1}$.

            \item Suppose $d_G(w) = 1$ at step $t$. 
            Let $z$ be the unique $G$ neighbor of $w$ at step $t$ that is not $u$.
        
            Suppose $z$ is not $G$ at step $t+1$.
            Then a different neighbor must flip to $G$ at step $t$.
            For each neighbor $z' \in N(w) \setminus \{u\}$, there are at most one clusters that contain $z'$ and flip it to $G$. 
            Since there are at most $\Delta$ neighbors of $w$, one of at most $\Delta$ clusters must flip for $G \in C^2_{t+1}$ in this case.
        
            Now suppose $z$ is $G$ at step $t+1$.
            Since $G \in C^0_{t}$, it must be that there exists a neighbor $x \in N(z) \setminus\{w,u,v^*\}$ such that $X_{t}(x) = Y_{t}(x) \in \{R,B\}$.
            For $G \in C^2_{t+1}$, $x$ must flip to a color that is not $R$ or not $B$. 
            Hence, there are at most $k-1$ clusters that contain $x$ and flip $x$ to not $R$ or not $B$. 
            Thus, one of at most $k$ clusters must flip for $G \in C^2_{t+1}$ in this case.
            
            \item Suppose $d_G(w) \geq 2$ at step $t$. 
            For $G \in C^2_{t+1}$, there must be a $G$ neighbor of $w$ that flips to not $G$ at step $t$.
            Fix any $z_1, z_2 \in N(w)$ distinct $G$ neighbors of $w$ at step $t$.
            It follows that $z_1$ or $z_2$ must flip to not $G$. 
            For each $z_i \in \{z_1,z_2\}$, there are at most $k-1$ clusters that contain~$z_i$ and flip it to a color that is not $G$.
            Thus, one of at most $2k$ clusters must flip for $G \in C^2_{t+1}$ in this case.
        \end{enumerate}
        
        \item Suppose $d_R(u) = 2$  or $d_B(u) = 2$ at step $t$.
        Let $w_1$ and $w_2$ be the $R$ or $B$ neighbors of $u$ at step $t$ that is not $v^*$.
    
        Suppose $w_1$ or $w_2$ is not $R$ or $B$ at step $t+1$.
        Then a different neighbor must flip to $R$ or $B$ at step $t$.
        For each neighbor $w' \in N(u) \setminus \{v^*\}$, there are at most two clusters that contain $w'$ and flip it to $R$ or $B$. 
        Since there are at most $\Delta$ neighbors of $u$, one of at most $\Delta$ clusters must flip for $G \in C^2_{t+1}$ in this case.
    
        Now suppose $w_1$ and $w_2$ are $R$ or $B$ at step $t+1$.
        Since $G \in C^0_{t}$, it must be that there exists a neighbor $z \in N(w_1) \cup N(w_2) \setminus\{u,v^*\}$ such that $X_{t}(z) = Y_{t}(z) \in \{R,B\}$.
        If $G \in C^2_{t+1}$, $z$ must flip to a color that is not in $\{R,B\}$. 
        Hence, there are at most $k-1$ clusters that contain $z$ and flip $z$ to not $R$ or $B$. 
        Thus, one of at most $k$ clusters must flip for $G \in C^2_{t+1}$ in this case.

        \item Suppose $d_R(u) = 1$  and $d_B(u) = 1$ at step $t+1$.
        For $G \in C^2_{t+1}$, a $R$ or $B$ neighbor of $u$ (that is not $v^*$) must flip to not $R$ or $B$ at step $t$.
        Fix any $w_r, w_b \in N(u) \setminus \{v^*\}$ distinct $R$ or $B$ neighbors respectively of $u$ at step $t$.
        It follows that $w_r$ or $w_b$ must flip to not $R$ or $B$ respectively. 
        For each $w_i \in \{w_r,w_b\}$, $k-1$ clusters contain $w_i$ and flip it to a color that is not $R$ or $B$ respectively.
        Thus, one of at most $2k$ clusters must flip for $G \in C^2_{t+1}$ in this case.
        
        \item Suppose $d_R(u) + d_B(u) \geq 3$ a step $t$.
        For $G \in C^2_{t+1}$, a $R$ or $B$ neighbor of $u$ (that is not $v^*$) must flip to not $R$ or $B$ at step $t$.
        Fix any $w_1, w_2, w_3 \in N(u) \setminus \{v^*\}$ distinct $R$ or $B$ neighbors of $u$ at step $t$.
        It follows that $w_1$, $w_2$, or $w_3$ must flip to not $R$ or $B$. 
        For each $w_i \in \{w_1,w_2,w_3\}$, there are at most $k-2$ clusters that contain $w_i$ and flip it to a color that is not $R$ or $B$.
        Thus, one of at most $3k$ clusters must flip for $G \in C^2_{t+1}$ in this case.
    \end{enumerate}
     
    \item Suppose $d_G \geq 3$ at step $t$. 
    Then $\geq d_G - 2$ neighbors of $v^*$ recolor from $G$ to a different color at time $t$.
    Fix any $u_1, u_2,u_3 \in N(v^*)$ distinct $G$ neighbors of $v^*$ at step $t$.  Each $u_i\in \{u_1,u_2,u_3\}$ has $\leq k$ clusters whose flip recolors $u_i$, and hence  $E_{0,2}(c)\leq 3k$ in this case.
\end{enumerate}

Again, we finish by considering the case when $v^*$ changes (i.e., $X_{t+1}(v^*) \neq R$ or $Y_{t+1}(v^*) \neq B$).
For $G \in C^2_{t+1}$ then $X_{t+1}(v^*) \neq Y_{t+1}(v^*)$.
It follows that in exactly one chain, a cluster containing $v^*$ and one of its neighbors must flip.
That is, $v^*$ must flip to the color of a neighbor in one chain. 
Since $v^*$ has $\Delta$ neighbors, there are at most $\Delta$ clusters that flip $v^*$ to a color in its neighborhood in one chain. 
Thus, one of at most $2\Delta$ clusters must flip for $G \in C^2_{t+1}$ in this case.

Finally, we want to upper bound $E_{1,2}(c)$.
Suppose $G \in C^1_{t}$.
Then there exists a unique $u \in N(v^*)$ such that $X_t(u) = Y_t(u) = G$. 
Moreover, there exists a unique $w \in N(u) \setminus \{u,v^*\}$ such that $X_t(w) = Y_t(w)$ is $R$ or $B$.
Suppose without loss of generality that $X_t(w) = Y_t(w)= R$.
We will show that for $G \in C^2_{t+1}$ either (i) one of $\Delta$ clusters of size $\geq2$ must flip or (ii) one of $\Delta$ clusters of size $1$ must flip and one of at most $4\Delta$ clusters must flip.
Hence, we will show the following:
\begin{equation}
    \label{eqn:lastone}
E_{1,2}(c)\leq P_2\Delta\alpha +  4\Delta\alpha^2 \leq \alpha \Delta/2
\end{equation}
where the last inequality holds because $P_2 = 185/616$ and $4\alpha
\Delta = \eps^*\Delta/(1250k) < 1/616$ since $k>\Delta$. 

It remains to prove \cref{eqn:lastone}.  Since $G \in C^2_{t+1}$, $d_G = 2$ at step $t+1$ and it follows that at least one vertex $u' \notin N(v^*)\setminus \{u\}$ must flip to $G$. 
At most, $\Delta$ neighbors of $v^*$ are not $u$; thus, at most, $\Delta$ clusters could flip a neighbor of $v^*$ to $G$, this is case (i).

Suppose a cluster of size $1$ flips at step $t$ and changes a neighbor $u' \in N(v^*) \setminus \{u\}$ to $G$. 
If $u$ is not $G$ at step $t+1$ or if there exists a $R,G$ path between $u$ and $u'$ that doesn't go through $v^*$, then another neighbor of $v^*$ must flip to $G$ at step $t$. 
At most, $\Delta$ clusters flip a neighbor of $v^*$ to $G$.
If neither of those two events occurs, then a neighbor of $u$ must change $R$ or $B$, or a neighbor of $w$ that is not a $u$ or $u'$ must flip to $G$. 
In all cases, at least one of $\Delta + 2\Delta + \Delta = 4\Delta$ clusters must flip for $G \in C^2_{t+1}$.  This completes the proof of \cref{eqn:lastone} and thus the proof of the lemma.
\end{proof}

\subsection{Combining~$\nabla_H, \nabla_B$ to Prove~\cref{lem:whamcontractlocal}}
\begin{proof}[Proof of \cref{lem:whamcontractlocal}]
Recall, our goal is to show that there exists $\beta>0$ such that
\[
\ExpCond{\wHamming(X_{t+1}, Y_{t+1})}{X_t, Y_t} \leq (1 - \beta) \wHamming(X_t, Y_t).
\]
Recall $\eps^* = 10^{-7}$, $\alpha = \varepsilon^*/5000k$, and $\delta = \frac{11}{6} - \frac{161}{88} = \frac{1}{264}$. 
We define $\eta = \delta \Delta / (34 k)$.
Note that if $ k \geq \Delta$ then $0 < \eta < 1/2$  as needed.

First, we bound $\nabla_B$ by combining \cref{lem:barnablocalneg,lem:barnablocalpos}. 
From \cref{lem:barnablocalneg} we get the following bound:
\begin{align}
    \frac{\eta}{ \Delta} \sum_{c\in C} \left[E_{1,0}(c) + 2E_{2,0}(c) \right] 
    &\geq \frac{\eta}{\Delta} \left(\alpha(k-\Delta-2)\gamma(v^*)(1 - 20k\alpha)^3 \right) &&\mbox{(by \cref{lem:barnablocalneg})} \nonumber\\ 
    &\geq \alpha \eta \frac{99}{100} \left(\frac{k}{\Delta}-1 - \frac{2}{\Delta}\right) \gamma(v^*)  &&\mbox{($(1 - 20k\alpha)^3 \geq 99/100$ for $\eps^*\leq 1/10$)} \nonumber \\  
    &\geq \alpha \eta \frac{99}{100}\left(\frac{k}{\Delta}- \frac{3}{2}\right) \gamma(v^*) &&\mbox{($\Delta \geq 4$)} \nonumber\\ 
    &\geq \alpha \delta\frac{99 \Delta}{3400 k} \left(\frac{k}{\Delta}- \frac{3}{2}\right) \gamma(v^*)  &&\mbox{($\eta = \delta \Delta/(34k)$)} \nonumber  \\
    &\geq \alpha \delta\frac{99 \Delta}{3400 k} \left(\frac{3}{2} + \frac{32}{99} - \frac{3}{2} \right) \gamma(v^*)  &&\mbox{($k/\Delta \geq \frac{11}{6} - \eps^* \geq \frac{3}{2} + \frac{32}{99}$)} \nonumber  \\
    &\geq \alpha \delta \frac{32 \Delta}{3400 k} \gamma(v^*). \label{lem:barnneg}
\end{align}
Likewise, from \cref{lem:barnablocalpos} we get
\begin{align}
    \frac{\eta}{ \Delta} \sum_{c\in C} \left[E_{0,1}(c) + 2E_{0,2}(c) + E_{1,2} \right] &\leq \frac{\eta}{\Delta} \left(17\alpha k(1-\gamma(v^*)) + \alpha ({31}/{100}) \Delta \gamma(v^*)\right)  \nonumber \\
    &\leq \alpha \delta \Delta \left( \frac{1}{2} (1- \gamma(v^*)) + \frac{31\Delta }{3400 k} \gamma(v^*) \right).
    \label{lem:barnpos}
\end{align}

Plugging \cref{lem:barnneg,lem:barnpos} into \cref{eq:barnablocalneg} we get
\begin{align}
    \nabla_B & \geq \frac{\eta}{\Delta}\sum_{c\in C} \left[E_{1,0}(c) + 2E_{2,0}(c) - E_{1,2}(c)\right]  - \frac{\eta}{\Delta} \sum_{c \in C} \left[E_{0,1}(c) + 2E_{0,2}(c)\right] \nonumber
    \\
    &\geq \alpha \delta \Delta \left( \frac{32 \Delta}{3400 k} \gamma(v^*) -  \frac{1}{2} (1- \gamma(v^*)) - \frac{31\Delta }{3400 k} \gamma(v^*) \right)  \nonumber\\
    &= \alpha \delta \Delta \left( \frac{\Delta}{3400 k} \gamma(v^*) -  \frac{1}{2} (1- \gamma(v^*)) \right).  \label{eq:nab_bound}
\end{align}

Recall $\nabla_\mathcal{H} = \nabla_H - \nabla_B$. 
We can combine \cref{lem-improved:disagree-effect,eq:nab_bound} to get
the following:
\begin{align*}
    \nabla_\mathcal{H} &= \nabla_H - \nabla_{B} && \\
& = - \Prob{X_{t+1}(v^*)=Y_{t+1}(v^*)} 
+ \sum_{c:d_c(v^*)>0}\Exp{H_c(X_{t+1},Y_{t+1})} &&
\\ 
\nonumber
& \qquad  + \sum_{T:\distflip(v^*,T)= 2} 
|T|\Prob{X_{t+1}(T)\neq Y_{t+1}(T)} - \nabla_{B}
\\
\nonumber
& \leq  - \alpha|A(v^*)|(1-10k\alpha)
 + \sum_{c \in C^0:\atop d_c(v^*)\geq 1} \alpha\left(\left(\frac{11}{6}-\delta\right)d_c(v^*)-1\right)(1+400k\alpha) &&
 \\
\nonumber
& \qquad
+ \sum_{c \in C^1\cup C^2} \alpha\left(\frac{11}{6}d_c(v^*)-1\right)(1+400k\alpha) + \ 300\Delta^2\alpha^2
  - \alpha \delta \Delta \left( \frac{\Delta}{3400 k} \gamma(v^*) -  \frac{1}{2} (1- \gamma(v^*)) \right)
 \\
& \leq  - \alpha(k - (11/6)\Delta + \delta\Delta(1-\gamma(v^*))) 
+ 1110k^2\alpha^2 
- \alpha \delta \Delta \left( \frac{\Delta}{3400 k} \gamma(v^*) -  \frac{1}{2} (1- \gamma(v^*)) \right)
\\ 
& \leq 
\alpha \eps^*\Delta + 1110k^2\alpha^2
- \alpha \delta \Delta \left( \frac{\Delta}{3400 k} \gamma(v^*) +  \frac{1}{2} (1- \gamma(v^*)) \right)  
\hspace{+.5in} \mbox{(since $k>(11/6-\eps^*)\Delta$)}
\\
& \leq 
\alpha \Delta \left( \eps^* + 2220 \alpha k
- \frac{\delta }{6800}  \right) \hspace{+2.25in} 
 \mbox{(since $k < 2\Delta$ and $\gamma\leq 1$)} \\
& \leq 
\alpha \Delta \left( 2 \eps^*
- \frac{\delta}{6800}  \right) \hspace{+2.8in} 
 \mbox{(since $\eps^* = 5000\alpha k$)} \\
& \leq 
-\alpha \eps^* \Delta
\end{align*}
where the last inequality follows from $\delta = 1/264 \geq 13600 \eps^*$ since $\eps^* = 10^{-7}$.
Hence, if $k \geq (11/6 - \eps^*)\Delta$ then for $\eps^*=10^{-7}$ we have $\nabla_\mathcal{H} \leq -\alpha \Delta \eps^* = -(\eps^*)^2\Delta/(5000k) \leq -10^{-18}$ since we can assume $k<2\Delta$ (the case $k>2\Delta$ follows from \cref{thm:main-weaker} which was established in \cref{sec:weaker,sec:coupling}).
Finally, applying the path coupling~\cref{thm:path-coupling} as in the proof of \cref{thm:main-weaker} in \cref{sec:weaker} completes the proof.
\end{proof}

\bibliography{references.bib}
\end{document}